\documentclass[graybox, envcountchap]{svmult}

\let\cleardoublepage\clearpage

\usepackage{graphicx}   
\usepackage{multicol}     
\usepackage[bottom]{footmisc}
\usepackage{newtxtext} 
\usepackage{textcomp}
\usepackage{cite}
\usepackage{subcaption}
\usepackage{overpic}
\usepackage{array}
\usepackage{lipsum}  
\usepackage{amsmath}
\usepackage{amssymb}

\newcommand{\vect}[1]{\mathbf{#1}}

\def\Htran{\mbox{\tiny $\mathrm{H}$}}
\def\Ttran{\mbox{\tiny $\mathrm{T}$}}
\def\CN{\mathcal{N}_{\mathbb{C}}} 
\def\imagunit{\mathsf{j}} 
\def\Ptx{P_{\mathrm{tx}}}
\def\Prx{P_{\mathrm{rx}}}
\def\mod{\mathrm{mod}}

\begin{document}

\title{Near-Field Beamforming and Multiplexing Using Extremely Large Aperture Arrays}
\label{chap:ELAA}

\author{Parisa Ramezani and Emil Bj\"ornson}

\institute{Parisa Ramezani \at KTH Royal Institute of Technology, Electrum 229, 16440 Kista, Sweden \email{parram@kth.se}
\and Emil Bj\"ornson \at KTH Royal Institute of Technology, Electrum 229, 16440 Kista, Sweden \email{emilbjo@kth.se}}

\maketitle

\abstract{The number of users that can be spatially multiplexed by a wireless access point depends on the aperture of its antenna array. When the aperture increases and wavelength shrinks, ``new'' electromagnetic phenomena can be utilized to further enhance network capacity. In this chapter, we describe how extremely large aperture arrays (ELAA) can extend the radiative near-field region to kilometer distances. We demonstrate how this affects the propagation models in line-of-sight (LoS) scenarios and enables finite-depth beamforming. In particular, it becomes possible to simultaneously serve users that are located in the same direction but at different distances.}

\section{Introduction}

The access points in current wireless networks use arrays of antennas for beamforming and spatial multiplexing. The former refers to the spatial focusing of each radiated signal on its intended receiver, while the latter refers to the simultaneous beamformed transmission of different data to users located at different locations; see Fig.~\ref{fig:mu-MIMO-basic}.
This technology is called massive multiple-input multiple-output (mMIMO) in 5G networks, which typically use arrays of 64 antennas in the 3 GHz band \cite{Bjornson1}, while the underlying communication theory for mMIMO supports arbitrarily many antennas and any radio-frequency band \cite{massivemimobook,Marzetta,Larsson}.
When the research community is looking beyond the 5G mMIMO scenarios, new terminologies are being used to distinguish the new characteristics.
The term \emph{ultra mMIMO} is used in \cite{Akyildiz,Faisal,Jamali} for systems with hundreds or even thousands of antennas operating in mmWave and THz bands. The terms \emph{large intelligent surfaces} and \emph{holographic mMIMO} \cite{Pizzo1,Huang,Dardari2021a,Pizzo2} are used when the arrays consist of electrically small and densely packed antennas, so that propagation effects can be studied using integrals over the aperture rather than summations of individual antennas.
In this chapter, we will consider the latter situation but focus on scenarios when the array aperture is so large that the users are in the so-called near-field. We will refer to such arrays as \emph{extremely large aperture arrays (ELAAs)} as in \cite{Bjornson1,Csatho2020a}. In this context, the aperture is measured as the array's length relative to the wavelength.

\begin{figure}[t!]
	\centering 
	\begin{overpic}[width=0.75\columnwidth,tics=10]{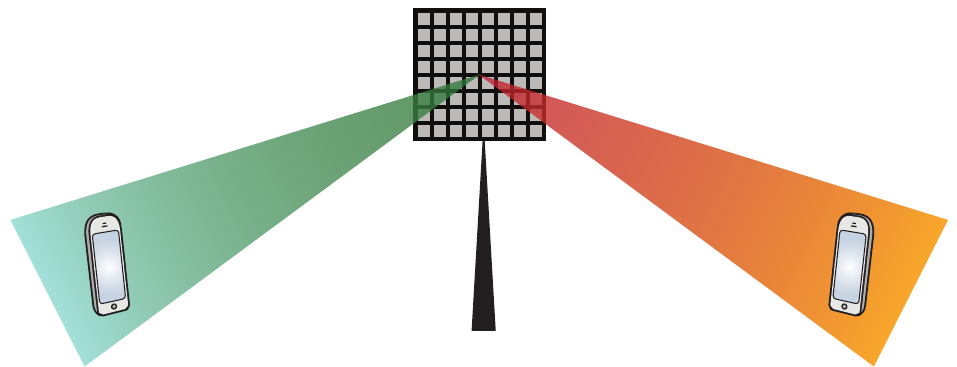}
	 \put (4,20) {\small Beam 1}
  \put (87,20) {\small Beam 2}
  \put (59,35) {\small 5G mMIMO array}
  \put (59,31) {\small with 64 antennas}
\end{overpic} 
	\caption{mMIMO is used for spatial multiplexing of user devices, which are served simultaneously using different beams. A 5G scenario with 64 antennas is shown.}
	\label{fig:mu-MIMO-basic} 
\end{figure}
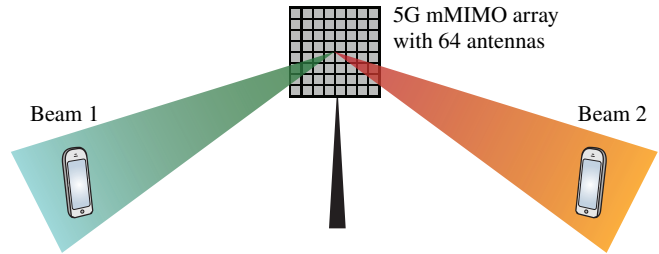

Traditionally, three regions are defined for antennas: the reactive near-field, the radiative near-field, and the far-field. These regions are defined from the transmitter perspective but can be equivalently viewed from the receiver due to reciprocity \cite{Selvan}. Consider the isotropic transmit antenna in Fig.~\ref{fig:three-fields}. Inductive coupling appears in the reactive near-field closest to the antenna and is commonly used by radio-frequency identification (RFID) tags, but the range is very short.
In this chapter, we are only concerned with the radiated electromagnetic waves and these have spherical wavefronts but can be approximated as planar in the far-field.
Conventional wireless communication systems operate in the far-field while the radiative near-field, where the spherical shape can be observed, is shorter than a meter from the transmitter.
However, this will change with the usage of ELAAs and it is therefore imperative to develop a theory for communications in the radiative near-field.

\begin{figure}[t!]
	\centering 
	\begin{overpic}[width=0.75\columnwidth,tics=10]{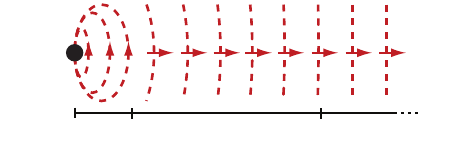}
	 \put (-3,18.7) {\small Transmitter}
  \put (17,3) {\small Reactive}
  \put (17,0) {\small near-field}
  \put (36,0) {\small Radiative near-field}
  \put (74,0) {\small Far-field}
\end{overpic} 
	\caption{There are three field regions around a transmitter. Conventional systems operate in the far-field where the waves are approximately planar. When using an ELAA, the users might be in the radiative near-field where the waves are spherical.}
	\label{fig:three-fields} \vspace{-3mm}
\end{figure}
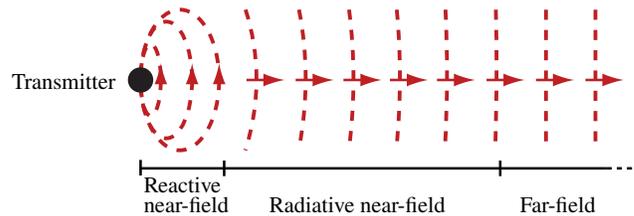

This chapter aims to lay the foundations of near-field communication with ELAAs by inspecting the unique characteristics of communicating in the radiative near-field. Specifically, three important properties must be considered when studying the near-field behavior, namely, the differences in the distances to the individual antennas, in their effective antenna areas, and in their polarization losses. We first derive a closed-form expression for the line-of-sight (LoS) channel gain between a single-antenna device and a planar array based on these three features and show that the classical far-field channel gain as well as the models in \cite{Hu,Bjornson3,Garcia,Tang,Ellingson} that do not take all the  near-field properties into account become inaccurate as the size of the ELAA grows large.
We then use the near-field compliant LoS formula to study the asymptotic signal-to-noise ratio (SNR) limit, as the array size goes to infinity.
Building on the traditional definition of antenna gain, we then introduce the antenna array gain and define a new metric for characterizing the near- and far-field regions of an ELAA. Next, the physical shape of the beams created in the near-field is explored in terms of width and depth \cite{Bjornson4}. Finally, we demonstrate how the finite depth enables a new mode of spatial multiplexing where users are distinguished in the depth domain.

\section{Channel Gain Modeling in the Radiative Near-Field}

In this section, we will analyze the channel gain in free-space LoS propagation scenarios. We begin with a motivating example to explain the need for moving beyond classical far-field models.

Consider the LoS scenario in Fig.~\ref{fig:fs1}, where an ideal isotropic transmit antenna sends a signal to a planar receive antenna with area $A$ located at distance $d$. The classical Friis' transmission formula \cite{Friis} manifests that received power is
\begin{equation} \label{eq:received_power_SISO}
\Prx =  \frac{A}{4\pi d^2} \Ptx
\end{equation}
where $\Ptx$ denotes the transmit power and the factor
\begin{equation} \label{eq:beta-definition}
\beta_d=\frac{A}{4\pi d^2}
\end{equation}
is the channel gain (while its inverse $\beta_d^{-1}$ is the pathloss). The derivation of \eqref{eq:beta-definition} relies on several technical assumptions: the receive antenna is lossless, the incident wave is planar, the antenna is perpendicular to the wave propagation so that $A$ is the effective antenna area as well, and the antenna polarization matches perfectly with the wave.

\begin{example}{Example 1}If the receive antenna is isotropic, its area is $A= {\lambda^2}/{(4\pi)}$, where $\lambda$ is the wavelength.
When the carrier frequency is  $f=3$\,\textrm{GHz}, we have $\lambda = c/f = 0.1$\,\textrm{m} where $c = 3 \times 10^8$\,\textrm{m/s}  is the speed of light. For propagation distances $d\in [2.5, 250]$\,m, the channel gain $\beta_d$ ranges from $-50$\,dB to $-70$\,dB, which are very small numbers.
\end{example}

\begin{figure}[t!]
\vspace{-4mm}
   \begin{subfigure}[t]{0.32\textwidth}
   
   \begin{overpic}[width=\textwidth]{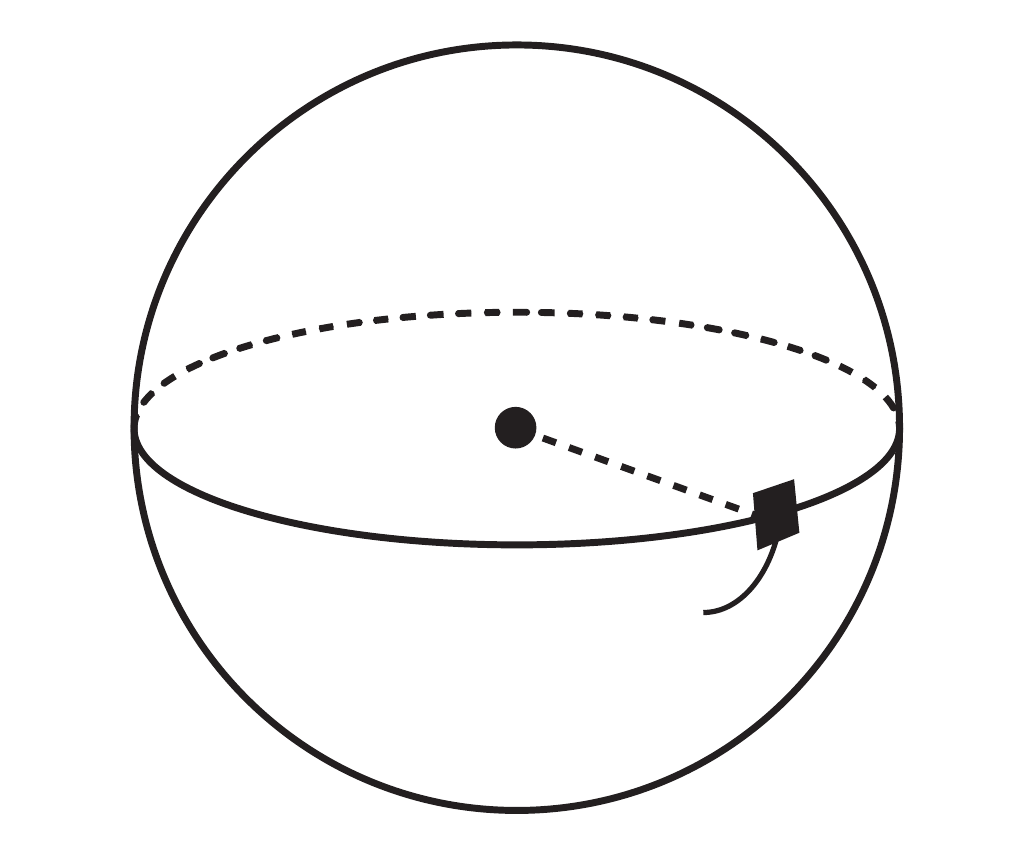}
        \put (16.5,41) {\scriptsize Transmitter}
 \put (59,41) {\small $d$}
  \put (28,20) {\scriptsize Receive antenna}
     \put (28,14) {\scriptsize with area $A$}
      \end{overpic} \vspace{-5mm}
     \caption{One receive antenna with area $A$.}
     \label{fig:fs1}
   \end{subfigure}
    \hspace{0.012\textwidth} 
   \begin{subfigure}[t]{0.32\textwidth}
  
   \begin{overpic}[width=\textwidth,tics=10]{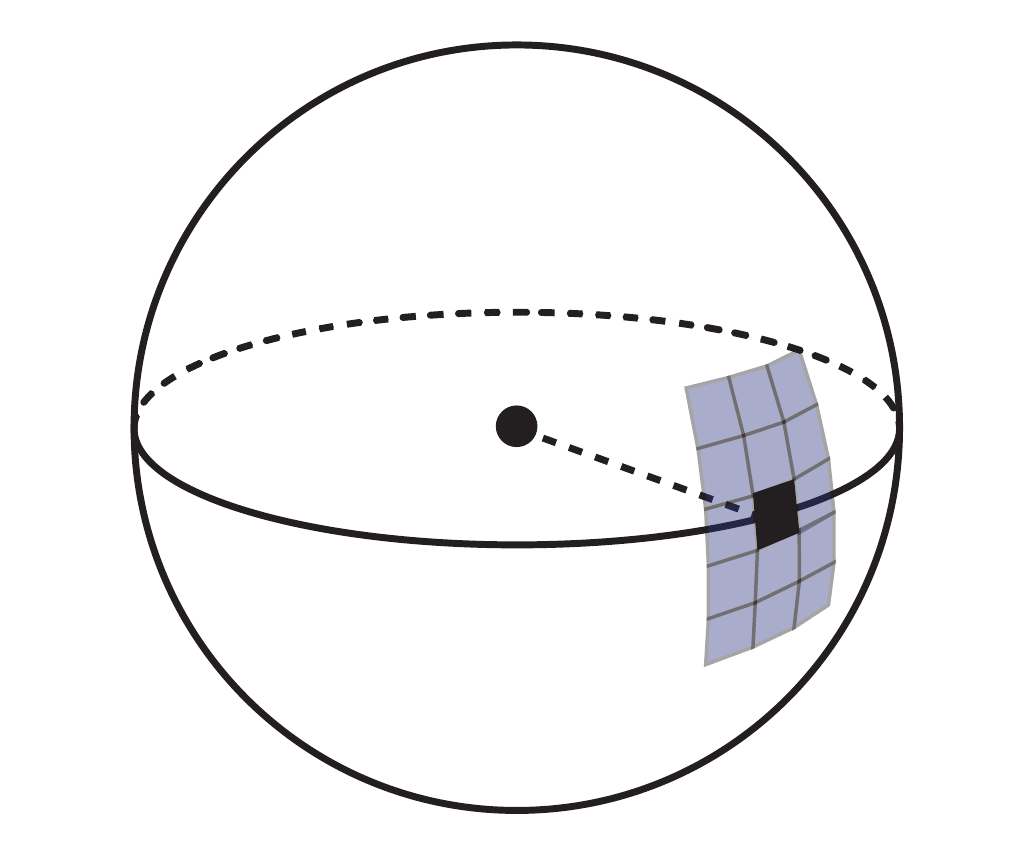}
  \put (16.5,41) {\scriptsize Transmitter}
 \put (59,41) {\small $d$}
 \put (26,20) {\scriptsize Spherical array}
     \put (28,14) {\scriptsize with $N$ antennas}
\end{overpic} \vspace{-5mm}
     \caption{Spherical array with $N$ equal-sized receive antennas.}
     \label{fig:fs2}
   \end{subfigure}
     \hspace{0.012\textwidth}  
   \begin{subfigure}[t]{0.32\textwidth}
   
   	\begin{overpic}[width=\textwidth,tics=10]{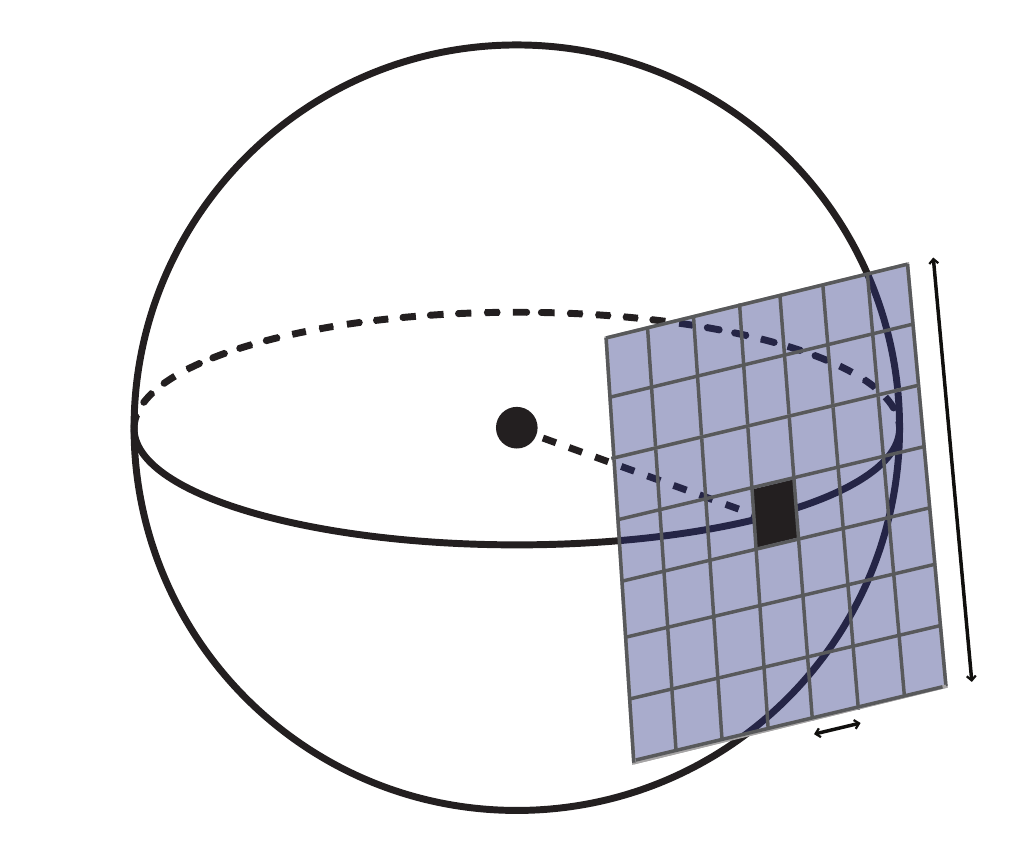}
  \put (16.5,41) {\scriptsize Transmitter}
 \put (53,34) {\small $d$}
 \put (27,22) {\scriptsize Planar array}
 \put (33,16) {\scriptsize with $N$}
     \put (34,10) {\scriptsize antennas}
   \put (93,36) {\small $\sqrt{NA}$}
   \put (77,6) {\scriptsize $\sqrt{A}$}
\end{overpic}

     \caption{Planar array with $\sqrt{N} \times \sqrt{N}$ equal-sized receive antennas.} \vspace{-5mm}
     \label{fig:fs3} 
   \end{subfigure}
   \caption{Examples of basic antenna scenarios.}
   \label{fig:fs}
\end{figure}
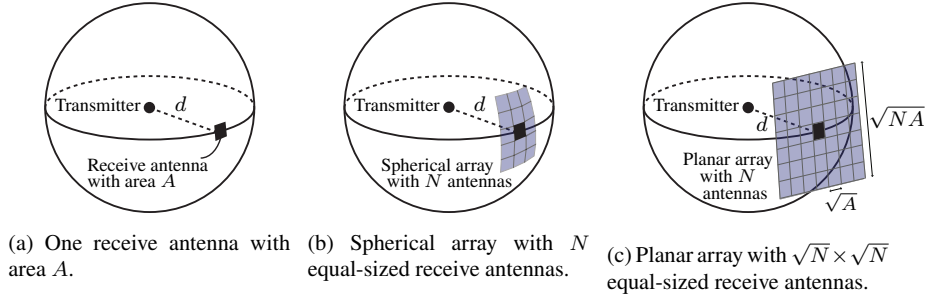

A way to increase the channel gain in \eqref{eq:beta-definition} is to make the total receive antenna area larger. This can be achieved by deploying $N$ antennas
(of the same kind as before) at the same distance, as illustrated in Fig.~\ref{fig:fs2}. If each antenna has an orientation and polarization that match its received signal, the total received power is $N$ times the value in \eqref{eq:received_power_SISO}:
\begin{equation} \label{eq:received_power_SISO2}
\Prx^{\textrm{spheric-}N} = N \Prx = N \beta_d \Ptx.
\end{equation}
If we now let $N\to \infty$ in \eqref{eq:received_power_SISO2}, we observe that the received power is approaching infinity. This is impossible since we can never receive more power than what was transmitted.
The catch is that we will eventually have the sphere in Fig.~\ref{fig:fs2} filled with antennas and then we cannot increase $N$ anymore. We need $N \beta_d \leq 1$, thus we cannot increase $N$ beyond $\beta_d^{-1}$.
We could, however, build arbitrarily large antenna arrays if they are planar as in Fig.~\ref{fig:fs3}. The outermost antennas will then be further than $d$ from the transmitter, which implies that the received power will not grow linearly with $N$ as in \eqref{eq:received_power_SISO2}. We need to develop different formulas for that scenario.

\subsection{Near-Field Compliant Channel Gain Modeling}
\label{sec:channel_gain}

We will now develop an asymptotically accurate channel gain model for the planar array case in Fig.~\ref{fig:fs3}. For brevity, the following assumption regarding the geometry of the ELAA is used in the remainder of this chapter (without being explicitly stated at every place), but the fundamental behaviors we  uncover are general.

\begin{example}{Assumption 1}The ELAA is a planar array with $N$ antennas that each has area $A$. The antennas have size $\sqrt{A} \times \sqrt{A}$ and are equally spaced on an $\sqrt{N} \times \sqrt{N}$ grid in the $XY$-plane.\footnote{Throughout this chapter, $\sqrt{N}$ is an integer for simplicity, but most of the analytical results only require a quadratic planar array with dimension $\sqrt{NA} \times \sqrt{NA}$. For a given array area $NA$, we can always adapt $A$ so that $N$ becomes the square of an integer.} The antennas are deployed edge-to-edge,  thus the total array area  is $NA$.
\end{example}

We will often consider ELAA as the receiver but the formulas that we derive also hold when the transmitter and receiver switch roles.
The \emph{effective} area of each receive antenna depends on its geometric location and rotation, with respect to the direction of the transmitter. If the receive antenna is fully perpendicular to the direction of propagation, then the effective area equals $A$. In any other case, the effective area is smaller than $A$. 

When the ELAA is in the radiative near-field of the transmitter (or vice versa), three fundamental properties must be accounted for when considering the amplitude and phase of the impinging wave:

\begin{enumerate}

\item The distances to the antennas vary over the array; 

\item The effective antenna areas vary since the antennas are seen from different angles; 

\item The losses from polarization mismatch vary since the signals are received from different angles.

\end{enumerate}

Fig.~\ref{fig:geometric_setup} shows the assumed setup with an isotropic transmitter and a receiving ELAA. 
If we number the antennas from left to right, row by row from the top, then the center of the $n$th receive antenna has coordinate $\vect{p}_n=(x_n,y_n,0)$ given by
\begin{align} \label{eq:xn}
x_n &= - \frac{(\sqrt{N}-1)\sqrt{A}}{2} + \sqrt{A} \, \mod(n-1,\sqrt{N}), \\
y_n &= \frac{(\sqrt{N}-1)\sqrt{A}}{2} - \sqrt{A} \left\lfloor \frac{n-1}{\sqrt{N}} \right\rfloor,  \quad \quad \quad \quad  \textrm{for }\,  n=1,\ldots,N, \label{eq:yn}
\end{align}
where  $\mod(\cdot,\cdot)$ is the modulo operation and $\lfloor \cdot \rfloor$ rounds the argument to the closest smaller integer.
The following lemma (adapted from \cite{Bjornson2,Dardari}) provides a general way of computing the channel gains to any of the $N$ antennas of a planar array.

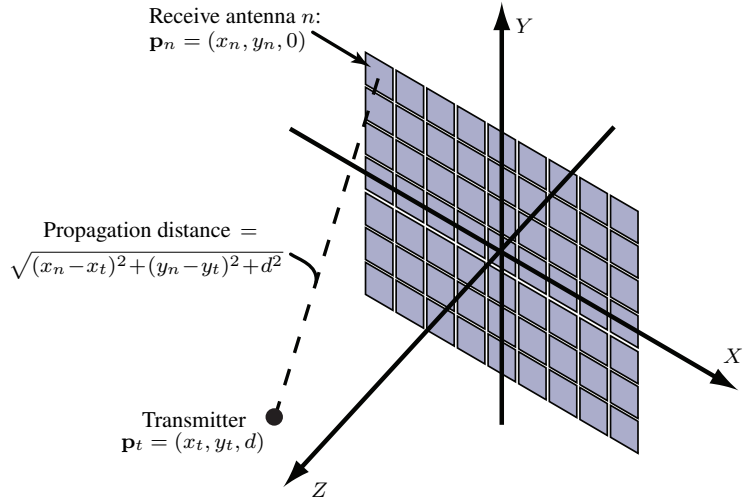
\begin{figure}[t!]
        \centering 
	\begin{overpic}[width=0.8\columnwidth,tics=10]{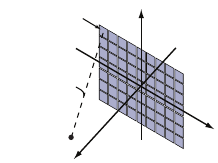}
	\put(14,10){Transmitter}
	\put(11,7){$\vect{p}_t=(x_t,y_t,d)$}
	\put(15,67.3){Receive antenna $n$:}
	\put(15,64){$\vect{p}_n=(x_n,y_n,0)$}
	\put(0,37){Propagation distance \,$=$}
	\put(-5,32){\small $\sqrt{(x_n\!-\!x_t)^2\!+\!(y_n\!-\!y_t)^2\!+\!d^2}$}
	\put(96.5,19){$X$}
	\put(67,66){$Y$}
	\put(38,0){$Z$}
\end{overpic} 
                \caption{An isotropic transmitter at an arbitrary location $\vect{p}_t=(x_t,y_t,d)$ transmits to a planar array located in the $XY$-plane. The  distance to receive antenna $n$ is shown.} 
                \label{fig:geometric_setup}
\end{figure}

\begin{lemma} \label{lemma1}
Consider a lossless isotropic antenna located at $\vect{p}_t=(x_t,y_t,d)$ that transmits a signal that has polarization in the $Y$-direction when traveling in the $Z$-direction.
 Suppose the planar receive antenna is located in the $XY$-plane, is centered at $\vect{p}_n=(x_n,y_n,0)$, and has area $a \times a$. The channel gain is given by
 
 \begin{align}
 \label{eq:total_gain}
\left|h_n(\vect{p}_t)\right|^2 &= \left|   \frac{1}{a} \int_{x_n-a/2}^{x_n+a/2} \int_{y_n-a/2}^{y_n+a/2} \epsilon(\vect{p}_t, \vect{p}_r) \partial x_r \partial y_r\right|^2
\end{align}
where $\vect{p}_r=(x_r,y_r,0)$ contains the integration variables and 
the impinging electric field is proportional to
\begin{equation}
\label{eq:complex_channel}
\epsilon(\vect{p}_t, \vect{p}_r) \!=\!  \frac{\sqrt{d \left((x_r - x_t)^2 + d^2\right)}}{\sqrt{4 \pi} \left((x_r - x_t)^2+(y_r - y_t)^2+d^2\right)^{\frac{5}{4}}} e^{-\imagunit \frac{2\pi}{\lambda}\sqrt{(x_r - x_t)^2+(y_r - y_t)^2+d^2}}.
\end{equation}
 \end{lemma}
 
 \begin{proof}
 The proof for this lemma is provided in Appendix \ref{app:proof_lemma1}.
 \end{proof}

The channel gain in \eqref{eq:total_gain} is computed as an integral of the electric field in \eqref{eq:complex_channel} over the antenna area. The signal phase should ideally be constant over the antenna so that the integration will coherently combine all the energy of the impinging signal. As illustrated in Fig.~\ref{fig:geometric_setup_matched_filtering}, the spherical wavefronts will create radial phase variations over the array, which are observable when the array is large compared to the propagation distance.
However, since the integral is computed on a per-antenna basis, we can mitigate this effect by creating ELAAs with sub-wavelength-sized antennas, so that the phase variation over each antenna is negligible (i.e., the antenna pattern is nearly isotropic). 
There will still be phase variations between the received signals at the different antennas, but these can be compensated for by digital receiver processing.
Under these conditions, we can achieve the following upper bound.

\begin{theorem} \label{lemma2}
The free-space LoS channel gain in \eqref{eq:total_gain} can be upper bounded as
\begin{align} \notag
&\left|h_n(\vect{p}_t)\right|^2 \leq \zeta_{\vect{p}_t,\vect{p}_n,a} = 
\\ &\frac{1}{4\pi} \sum_{x \in \mathcal{X}_{t,n}  } \sum_{ y \in \mathcal{Y}_{t,n}  } \left( \frac{\frac{xy}{d^2}}{3 \left(\frac{y^2}{d^2}+1\right)\sqrt{ \frac{x^2}{d^2}+\frac{y^2}{d^2}+1}} 
+ \frac{2}{3} \tan^{-1} \left(  \frac{\frac{xy}{d^2}}{\sqrt{ \frac{x^2}{d^2}+\frac{y^2}{d^2}+1}}
\right) \right) , \label{eq:channel-gain-general-case}
\end{align}
where $\mathcal{X}_{t,n} = \{ a/2+x_n-x_t,a/2-x_n+x_t \}$ and $\mathcal{Y}_{t,n}  = \{ a/2+y_n-y_t,a/2-y_n+y_t \}$.
\end{theorem}

\begin{proof}
The proof is given in Appendix~\ref{app:proof_lemma2}. 
\end{proof}

 \begin{figure}[t!]
        \centering 
	\begin{overpic}[width=0.8\columnwidth,tics=10]{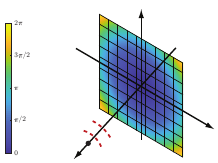}
	\put(24,8){Transmitter}
	\put(-3.5,67.3){Phase shifts}
	\put(96.5,19){$X$}
	\put(67,66){$Y$}
	\put(38,0){$Z$}
\end{overpic} 
                \caption{The spherical wavefronts create circular phase variations over the ELAA, which are substantial in the radiative near-field. This issue is addressed by using sub-wavelength-sized antennas so that the variation is negligible over each antenna.} 
                \label{fig:geometric_setup_matched_filtering}
\end{figure}
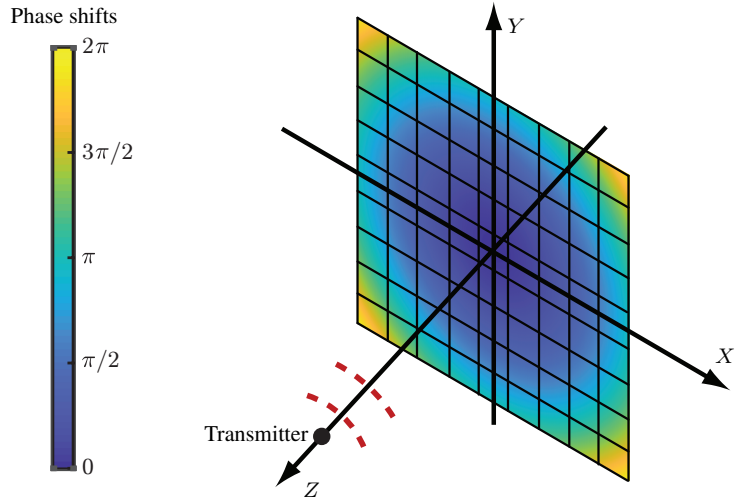

Theorem~\ref{lemma2} provides a closed-form upper bound on the channel gain by assuming negligible signal phase variations over the antenna area. The bound is tight for sub-wavelength-sized antennas, such as $a \le \lambda/4$  \cite{Bjornson2}, which will be assumed throughout this chapter.
There will also be amplitude variations between the antennas, which depend on the three aforementioned near-field properties. Their impact is not obvious from the channel gain formula in \eqref{eq:channel-gain-general-case}, but can be distinguished from \eqref{eq:pathloss-integral}.

The receiver can apply matched filtering to accumulate the received power of the $N$ antennas. The resulting total channel gain is $\sum_{n=1}^{N} \left|h_n(\vect{p}_t)\right|^2$, which we will show later in Section~\ref{subsec:system-model}.
If the transmitter is centered in front of the array, a compact formula  can be obtained as follows.

\begin{corollary} \label{cor:alpha_expression}
If the transmitter is centered in front of the planar array (i.e., $\vect{p}_t=(0,0,d)$), the received power is upper bounded by \vspace{-2mm}
\begin{equation} \label{eq:received_power_planar}
 \Prx^{\mathrm{planar}\textrm{-}N} = \sum_{n=1}^{N} \zeta_{\vect{p}_t,\vect{p}_n,a}  \Ptx = 
\alpha_{d,N} \Ptx,
\end{equation}
where the  total channel gain  (using
$\beta_d$ from \eqref{eq:beta-definition}) is
\begin{equation} \label{eq:alpha-expression} 
\alpha_{d,N} 
=   \frac{N \beta_d}{3(N \beta_d \pi+1) \sqrt{2 N \beta_d \pi + 1}} +  \frac{2}{3\pi}\tan^{-1} \!\left( \frac{N \beta_d \pi}{ \sqrt{2N \beta_d \pi + 1}} \right) .
\end{equation}
\end{corollary}

\begin{proof}
The upper bound is achieved by integrating the received power over the antenna array, which is equivalent to having one large antenna with no phase variations. Hence, we can obtain \eqref{eq:alpha-expression} from
Theorem~\ref{lemma2} by setting $x_t=y_t=0$, $x_n=y_n=0$, and $a=\sqrt{NA}$, in which case $\mathcal{X}_{t,n}  = \mathcal{Y}_{t,n}  = \{\sqrt{NA}/2,\sqrt{NA}/2\}$.
By replacing $d$ with $\sqrt{A/(4\pi \beta_d) }$ and rearranging the terms, we then obtain \eqref{eq:alpha-expression} from \eqref{eq:channel-gain-general-case}.
\end{proof}

The total channel gain in \eqref{eq:alpha-expression} 
is valid for arbitrarily large planar arrays and, particularly, supports communications in the radiative near-field. 

\subsection{Far-Field Approximation and Large-Array Limit}

The near-field compliant channel gain model that was derived above is also accurate in the far-field, thus there is no need to determine beforehand if the communication scenario corresponds to the near- or far-field.
We will show this by considering the two extremes: a small and a large array, compared to the propagation distance.

Suppose the planar array considered in Corollary~\ref{cor:alpha_expression} is in the far-field of the transmitter, in the sense that the distance is much larger than the array's diagonal: $d \gg \sqrt{2 N A}$.\footnote{The distance at which the far-field approximation of the total channel gain becomes accurate can be quantified based on the variations in the amplitude of the received signal over the array. This will be elaborated in Section \ref{sec:near-far-field} to get a more precise expression than  $d \gg \sqrt{2 N A}$.} In this case, $N\beta_d \pi +1 \approx 1$ and $ \sqrt{2N \beta_d \pi + 1} \approx 1$ in \eqref{eq:alpha-expression}.
By using the first-order Taylor approximation $\tan^{-1}(x) \approx x$, which is tight when the argument is close to zero (i.e., when $N \beta_d \pi$ is small), it follows from \eqref{eq:received_power_planar} that
\begin{equation} \label{eq:received_power_planar-approx}
\Prx^{\mathrm{planar}\textrm{-}N}  \approx \left( \frac{N \beta_d}{3} + \frac{2}{3\pi} N \beta_d \pi \right) \Ptx = N \beta_d \Ptx
\end{equation}
which is equal to $\Prx^{\textrm{spheric-}N}$ in  \eqref{eq:received_power_SISO2}. Hence,  the received power is proportional to $N$ for relatively small planar arrays, just as in the far-field. 

If $N$ grows large while $d$ is fixed, so that $d \ll \sqrt{2 N A}$, the far-field approximation is no longer valid.
We notice that as $N \to \infty$ it holds that
\begin{align}
 \frac{N \beta_d}{3(N \beta_d \pi+1) \sqrt{2 N \beta_d \pi + 1}}  &\to 0, \\
\tan^{-1} \!\left( \frac{N \beta_d \pi}{ \sqrt{2N \beta_d \pi + 1}} \right) &\to \frac{\pi}{2}.
\end{align}
 Hence, the received power in \eqref{eq:received_power_planar} saturates and has the asymptotic limit 
 \begin{equation} \label{eq:limit-planar-array}
 \Prx^{\mathrm{planar}\textrm{-}N} \to \frac{2}{3\pi} \frac{\pi}{2}  \Ptx
 = \frac{\Ptx}{3} \quad \textrm{as} \,\,N\to \infty. 
  \end{equation}
This asymptotic value is physically plausible since one third of the transmitted power is received. The reason that the limit is finite, although the array is infinitely large, is that each new receive antenna is deployed further away from the transmitter; the effective area (perpendicularly to the direction of propagation) becomes gradually smaller, and the polarization loss also increases. 

\begin{figure}[t!]
	\centering  \vspace{-3mm}
	\begin{overpic}[width=0.9\columnwidth,tics=10]{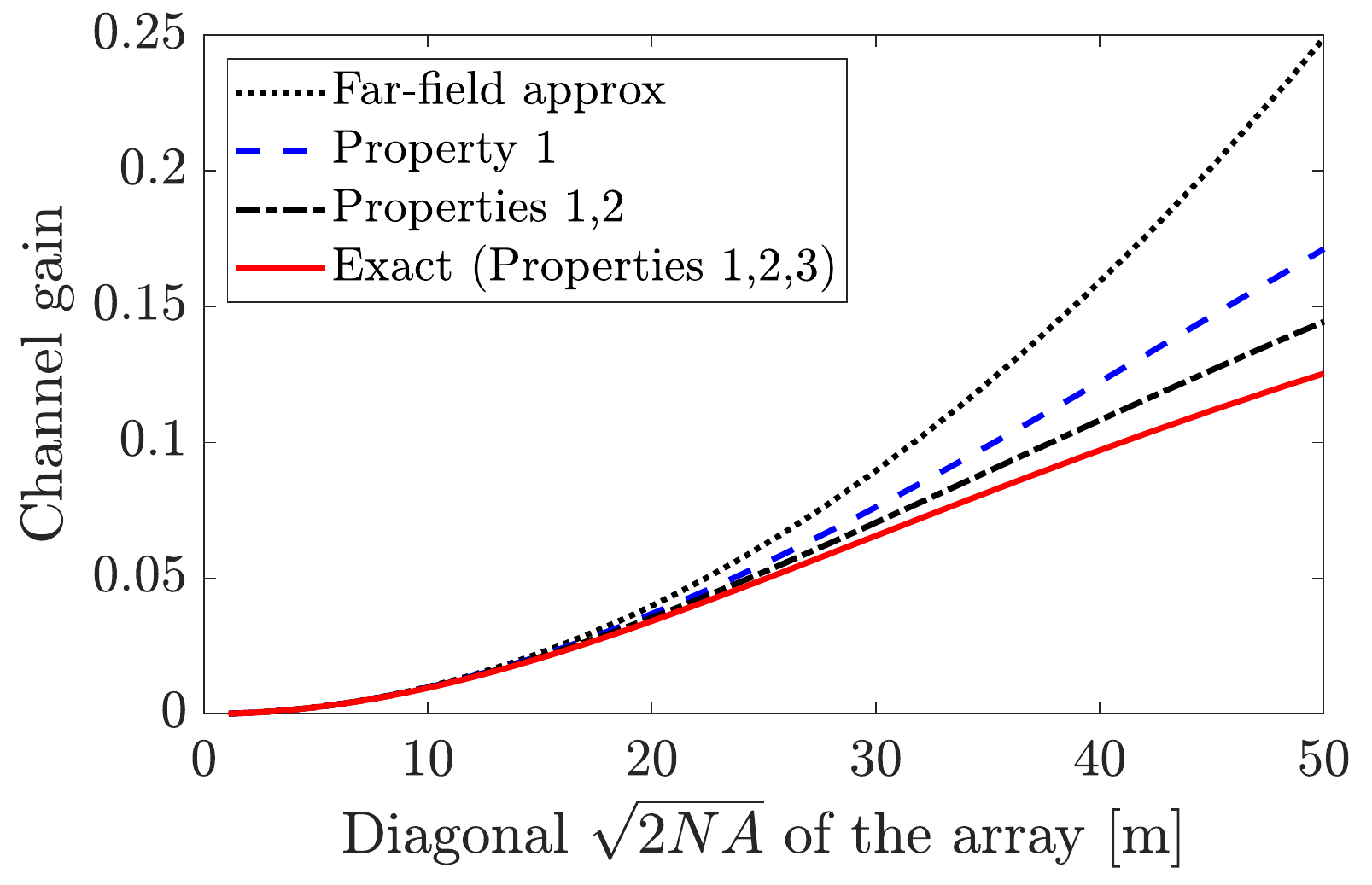}
\end{overpic} 
	\caption{The total channel gain $\alpha_{d,N}$ in \eqref{eq:alpha-expression} for a varying number of antennas, leading to a varying diagonal of the planar ELAA. 
	The exact curve is compared with approximations that are obtained by neglecting all or some of the essential  propagation properties in the radiative near-field.} \vspace{-1mm}
	\label{fig:near-field_modeling} 
\end{figure}

The literature contains alternative channel gain expressions designed for the radiative near-field. 
The models in \cite{Tang,Garcia,Ellingson} only capture the first near-field property (i.e., different distances to the antennas). The models in \cite{Hu,Bjornson3} also capture the second property (i.e., variation of effective antenna areas over the array).
By contrast, the exact model in \eqref{eq:channel-gain-general-case} and \eqref{eq:alpha-expression} also includes the polarization 
mismatch over the array.
Fig.~\ref{fig:near-field_modeling} shows the importance of including all three properties when studying the radiative near-field. The transmitter is $d=20$\,m in front of the array with $x_t=y_t=0$. We vary the number of antennas $N$ and plot the total channel gain in \eqref{eq:alpha-expression} as a function of the diagonal $\sqrt{2NA}$ of the array.
The exact formula and far-field approximation coincide when the diagonal is smaller than $d/2=10$\,m, but diverge for larger arrays.
If one only considers some of the near-field properties, the channel gain becomes overestimated. The gaps are rather small in this figure, but grow as $N\to \infty$. 
The channel gain converges to $1/2$ instead of $1/3$ if the polarization effects are neglected (``Properties 1,2'') \cite{Hu,Bjornson3}. The intuition is that an infinitely large array divides the world into two halves and, therefore, half of the power reaches the array. However, the polarization becomes increasingly mismatched for far-away antennas so one third of the total incident power cannot be received.
If also the variations in effective areas are neglected (``Property 1''), the channel gain diverges as $N\to \infty$ \cite{Bjornson2}. If we further assume that $\lambda=0.1$\,m (i.e., communication in the 3 GHz band) and $A=(\lambda/4)^2$, then the near-field formulas are required for $N > 10^5$.

In summary, many channel gain formulas are accurate for small and medium-sized arrays, but one must take all the three near-field properties into account when the array's diagonal is similar to or larger than the propagation distance.

\subsection{System Model for Uplink and Downlink}
\label{subsec:system-model}

Being equipped with a near-field compliant channel gain formula, we are ready to define the system model and study the achievable spectral efficiency (SE). 

\begin{example}{Assumption 2}The transmitter is at $\vect{p}_t= (d \sin(\varphi), 0, d \cos(\varphi))$ in the $XZ$-plane with distance $d$ from the array's center and angle $\varphi \in [-\pi/2,\pi/2]$, as illustrated in Fig.~\ref{fig:from_above}. It sends a signal that has polarization in the $Y$-direction when traveling in the $Z$-direction.
\end{example}

The flat-fading channel between the single-antenna transmitter and $N$-antenna receiver is represented by the vector $\vect{h} = [h_1,\ldots,h_N]^{\Ttran} \in \mathbb{C}^N$, where $h_n=|h_n| e^{-j\phi_n}$ is the channel from the transmitter to the $n$th receive antenna.
The channel gain $|h_n|^2 \in [0,1]$ can be computed using Theorem~\ref{lemma2}, while the phase-shift  $\phi_n \in [0,2\pi]$ can be computed based on the propagation delay as
\begin{align} 
\phi_n \!=\! 2\pi \cdot\mod \left( \! \frac{{||\vect{p}_t- \vect{p}_n||}}{\lambda}  ,1 \!\right) \!=\! 2\pi \cdot\mod \left( \!\frac{\sqrt{x_n^2+y_n^2+d^2 - 2d x_n \sin(\varphi)}}{\lambda}  ,1 \!\right) \!.
\end{align}
The received uplink signal $\vect{r} \in \mathbb{C}^N$ at the ELAA is   
\begin{equation}
\vect{r} = \vect{h} \sqrt{\Ptx} s +  \vect{n}  
\end{equation}
where $\Ptx$ is the transmit power, $s$ is the unit-norm information signal, and $\vect{n} \sim \CN(\vect{0},\sigma^2 \vect{I}_N)$ is the independent receiver noise with variance $\sigma^2$. Under the assumption of perfect channel knowledge, linear receiver processing is optimal \cite{Telatar,massivemimobook} and we let $\vect{v} \in \mathbb{C}^N$ denote the receive combining vector. 
It is well-known that the maximum SNR is achieved by matched filtering (MF) with $\vect{v} = \vect{h}/\|\vect{h} \|$ \cite{massivemimobook}, which is also known as maximum ratio combining. The SE is
\begin{equation} \label{eq:SE-mMIMO}
\log_2(1+\mathrm{SNR}_{\mathrm{MF}} ) \quad \textrm{[bit/s/Hz]}
\end{equation}
with 
\begin{equation} \label{eq:SNR-mMIMO}
\mathrm{SNR}_{\mathrm{MF}} =  \frac{|\vect{v}^{\Htran}\vect{h}|^2 \Ptx }{ \| \vect{v}\|^2 \sigma^2} = \|\vect{h} \|^2 \frac{ \Ptx}{\sigma^2}  =  \left(\sum_{n=1}^{N}|h_n|^2\right)\frac{ \Ptx}{\sigma^2}.
\end{equation} 

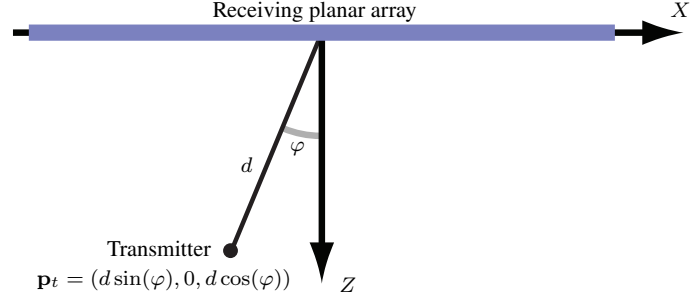
\begin{figure}[t!]
\centering
\begin{overpic}[width=.8\columnwidth,tics=10]{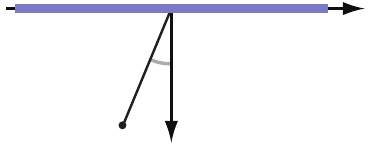}
	\put(30,40){Receiving planar array}
	\put(34,18){$d$}
	\put(41,21){$\varphi$}
	\put(48,1){$Z$}
	\put(95,40){$X$}
	\put(15,6){Transmitter}
	\put(5,2){$\vect{p}_t= (d \sin(\varphi), 0, d \cos(\varphi))$}
\end{overpic}
        \caption{Under Assumption 2, the transmitter is at distance $d$ from the center of the receiver and with  angle $\varphi$ in the $XZ$-plane.} 
                \label{fig:from_above} \vspace{-4mm}
\end{figure}

We can also consider the corresponding downlink scenario, where the ELAA transmits to a single-antenna isotropic receiver. Reciprocity implies that $\vect{h}$ is the channel vector also in this case.
We let $\Ptx$ denote the transmit power and $\vect{w} \in \mathbb{C}^N$ be the unit-norm linear precoding vector. The received downlink signal $r \in \mathbb{C}$ is 
\begin{equation}
    r = \vect{h}^{\Ttran} \vect{w} \sqrt{\Ptx} s + n,
\end{equation}
where $s$ is the unit-norm information signal and $n \sim \CN(0,\sigma^2 )$ is the independent receiver noise.
The corresponding downlink SNR is 
\begin{equation} \label{eq:SNR-mMIMO-downlink}
\frac{|\vect{h}^{\Ttran}\vect{w}|^2 \Ptx }{ \sigma^2} \leq \|\vect{h} \|^2 \frac{ \Ptx}{\sigma^2}  = \mathrm{SNR}_{\mathrm{MF}},
\end{equation} 
where the upper bound is achieved by the MF precoding $\vect{w}=\vect{h}^{\star} / \| \vect{h} \|$ \cite{massivemimobook} and $^{\star}$ denotes conjugation. It is also known as maximum ratio transmission. Importantly, the same SNR value $\mathrm{SNR}_{\mathrm{MF}}$ and SE $\log_2(1+\mathrm{SNR}_{\mathrm{MF}})$ are achieved in uplink and downlink, so we can study these cases jointly.

\subsection{SNR Expressions and Power Scaling Law}
\label{subsec:scaling_law}

We can compute the uplink/downlink SNR and channel gain numerically using Lemma~\ref{lemma1}.
When the array consists of physically small antennas, we can use Theorem~\ref{lemma2} to compute the SNR in closed form as follows.

\begin{theorem} \label{prop:SNR-mMIMO2}
When having small antennas, the SNR in \eqref{eq:SNR-mMIMO} with MF becomes
\begin{equation} \label{eq:SNR_mMIMO_general}
\mathrm{SNR}_{\mathrm{MF}}  =  \xi_{d,\varphi,N}\frac{ \Ptx}{\sigma^2} 
\end{equation}
where the total channel gain $\xi_{d,\varphi,N} $ is given by
\begin{align} \notag
\xi_{d,\varphi,N} =& \sum_{i=1}^{2} \Bigg( \frac{ B +(-1)^i \sqrt{B} \tan(\varphi)  }{6 \pi (B+1)\sqrt{ 2B+\tan^2(\varphi) + 1 + 2(-1)^i \sqrt{B} \tan(\varphi)} } \\ &+ \frac{1}{3\pi} \tan^{-1} \Bigg(  \frac{ 
B +(-1)^i \sqrt{B} \tan(\varphi) 
  }{
  \sqrt{2B+\tan^2(\varphi) + 1 + 2(-1)^i \sqrt{B} \tan(\varphi)}
  }
\Bigg) \Bigg) \label{eq:xi-mMIMO}
\end{align}
and $B = N \pi \beta_{d \cos(\varphi)} = \frac{NA}{4 d^2 \cos^2(\varphi)}$.

\end{theorem}
\begin{proof}
This result follows from Theorem~\ref{lemma2} with $\vect{p}_t=(d \sin(\varphi), 0, d \cos(\varphi))$, $\vect{p}_n=(0,0,0)$, and $a = \sqrt{NA}$.
\end{proof}

Note that the channel gain in \eqref{eq:xi-mMIMO} depends only on the total array area $NA$, thus the choice of frequency band only affects how many antennas are needed to achieve that area.
By using Corollary~\ref{cor:alpha_expression}, a more compact expression can be obtained for the case where the transmitter is centered in front of the array (i.e., $\varphi=0$).

\begin{corollary} \label{cor:SNR_mMIMO}
When the transmitter has angle $\varphi=0$, the SNR in \eqref{eq:SNR_mMIMO_general} simplifies to 
\begin{equation} \label{eq:SNR-mMIMO_planar}
\mathrm{SNR}_{\mathrm{MF}} = \alpha_{d,N}\frac{\Ptx}{\sigma^2}
\end{equation}
where the total channel gain $\alpha_{d,N}$ is given in \eqref{eq:alpha-expression}.
\end{corollary}

We  use the expression in Theorem~\ref{prop:SNR-mMIMO2} for an arbitrary $\varphi$ to study the far-field behavior in the next corollary. Note that $d \cos(\varphi)$ is the distance from the transmitter to the plane where the array is deployed and $\sqrt{2 N A}$ is the array's diagonal.

\begin{corollary}[Far-field approximation]If the transmitter is in the far-field of the ELAA, in the sense that $d \cos(\varphi) \gg \sqrt{2 N A}$, then \eqref{eq:SNR-mMIMO_planar} is well approximated as
\begin{equation} \label{eq:SNR-mMIMO_planar-farfield}
\mathrm{SNR}_{\mathrm{MF}} \approx \mathrm{SNR}_{\mathrm{MF}}^{\mathrm{ff}}  = N  \varsigma_{d,\varphi} \frac{ \Ptx}{\sigma^2} 
\end{equation}
where 
\begin{equation} \label{eq:updated-beta}
\varsigma_{d,\varphi} =
\beta_{d \cos(\varphi)} \cos^3(\varphi)
\end{equation}
and $\beta_{d \cos(\varphi)}$ is the same as in \eqref{eq:beta-definition} but with the distance $d \cos(\varphi)$.
\label{cor:far-field_mMIMO}
\end{corollary}
\begin{proof}
The derivation can be found in Appendix~\ref{app:proof_cor:far-field_mMIMO}.
\end{proof}

From Corollary~\ref{cor:far-field_mMIMO}, we notice that the far-field SNR in \eqref{eq:SNR-mMIMO_planar-farfield} is proportional to $N$. Hence, when  $N$ increases, the system can either benefit from a linearly increasing SNR or reduce $\Ptx$ as $1/N$ to keep the SNR constant. The latter is the classical \emph{power scaling law} for mMIMO that first appeared in \cite{Ngo,Hoydis} and has since appeared in numerous papers.
However, when computing the asymptotic behavior as $N \to \infty$, these seminal works implicitly assume the transmitter is always in the far-field and, thus, that the SNR grows linearly towards infinity as $N \to \infty$ for any fixed transmit power.
This is not physically possible because the receiver will have to receive more power than was transmitted. As $N$ increases, the transmitter will eventually be in the near-field of the ELAA, and then the total channel gain will saturate according to the formula in Theorem~\ref{lemma2}. The following is an accurate but less encouraging result.

\begin{corollary}[Asymptotic power scaling law] As $N \to \infty$ with a constant transmit power $\Ptx$, the SNR with MF satisfies
\begin{equation} \label{eq:limit-mMIMO}
\mathrm{SNR}_{\mathrm{MF}} \to \frac{1}{3}\frac{\Ptx}{ \sigma^2}.
\end{equation}
If the transmit power is reduced with $N$ according to the power scaling law $\Ptx = P/N^{\rho}$ for some constant $P>0$ and exponent $\rho>0$, then it holds that
\begin{equation} \label{eq:limit-mMIMO2}
\mathrm{SNR}_{\mathrm{MF}} = \xi_{d,\varphi,N} \frac{ P}{\sigma^2 N^{\rho}} \to 0, \quad \textrm{as }\, N \to \infty.
\end{equation}
\label{cor:limit_mMIMO}
\end{corollary}
\begin{proof}
The limit in \eqref{eq:limit-mMIMO} is computed similarly to the finite limit in \eqref{eq:limit-planar-array}. The result in \eqref{eq:limit-mMIMO2} follows from that $P \xi_{d,\varphi,N} $ has a finite limit and $1/N^{\rho}\to 0$ as $N \to \infty$.
\end{proof}

\begin{figure}[t!]
	\centering 
	\begin{overpic}[width=0.9\columnwidth,tics=10]{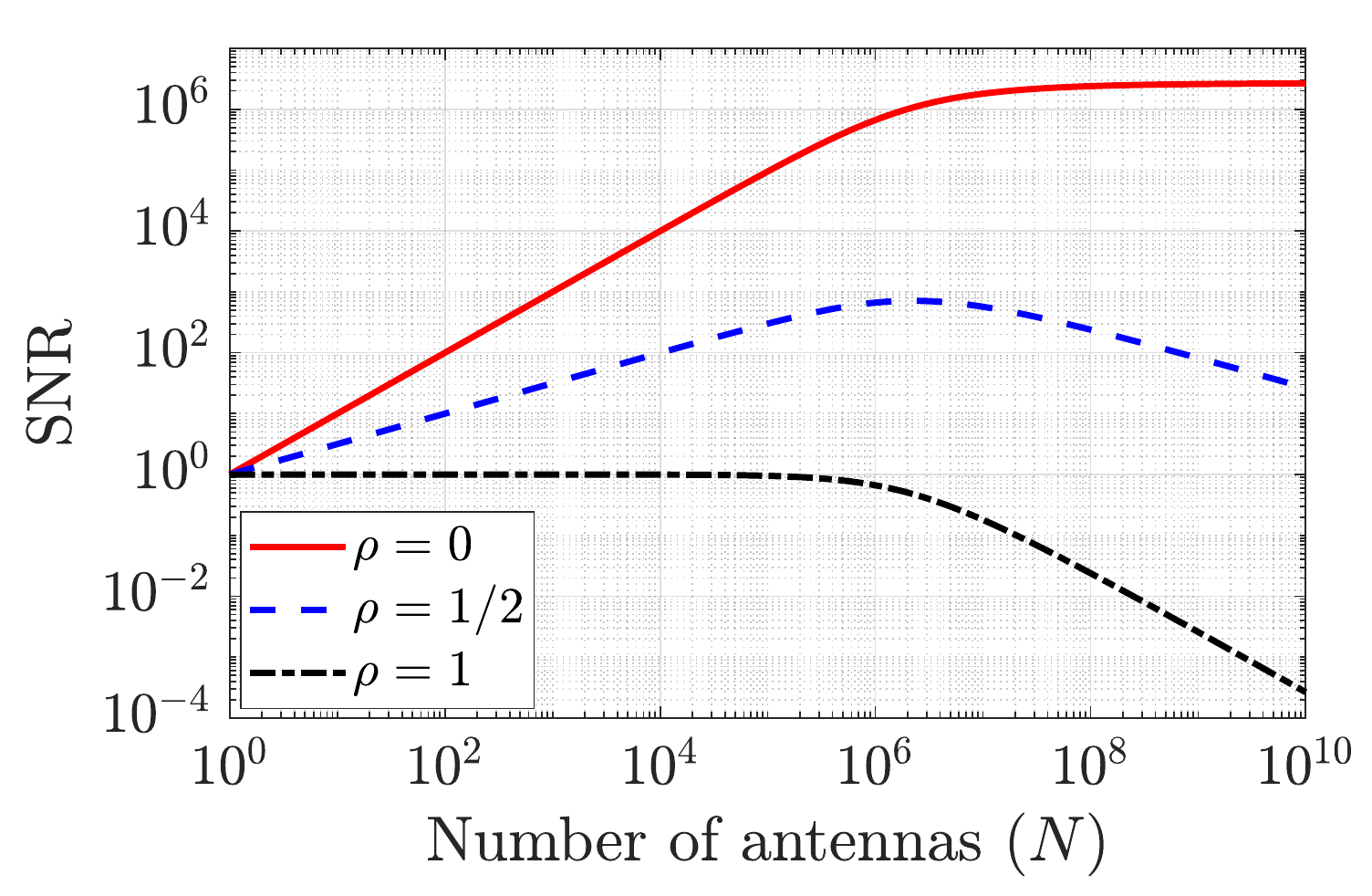}
\end{overpic} 
	\caption{The SNR value $\mathrm{SNR}_{\mathrm{MF}} $ in \eqref{eq:SNR-mMIMO_planar} when scaling down the transmit power as $\Ptx = P/N^{\rho}$ for $\rho \in \{0, 1/2, 1\}$.}
	\label{fig:simulation_scaling_law}  
\end{figure}

This corollary proves that any power scaling law will lead to zero SNR asymptotically. However, the scaling laws in previous literature will anyway accurately predict the SNR behavior in the far-field. We will demonstrate that by an example.

Fig.~\ref{fig:simulation_scaling_law} shows $\mathrm{SNR}_{\mathrm{MF}} $ in \eqref{eq:SNR-mMIMO_planar} when we scale down the transmit power as $\Ptx = P/N^{\rho}$ for $\rho \in \{0, 1/2, 1\}$, where $\rho=0$ corresponds to a constant power. The transmit power is selected so that $P \xi_{d,\varphi,N}/\sigma^2=0$\,dB for $N=1$. 
The simulation setup is the same as in Fig.~\ref{fig:near-field_modeling} with the wavelength $\lambda=0.1$\,m (i.e., $f=3$\,GHz) and each antenna has area $A = (\lambda/4)^2$.
We observe that for $\rho=0$, the far-field behavior of having an SNR that grows linearly with $N$ holds for any $N \le 10^5$, as also observed in Fig.~\ref{fig:near-field_modeling}.
If we select $\rho=1/2$, the SNR will instead grow as $\sqrt{N}$ for $N \le 10^5$. If $\rho=1$, the SNR is approximately constant for $N \le 10^5$.
For larger values of $N$, the SNR goes to zero whenever $\rho>0$, as proved analytically in Corollary~\ref{cor:limit_mMIMO}.

Since this example considers $\varphi=0$, we know from Corollary~\ref{cor:SNR_mMIMO} that $\xi_{d,0,N} =  \alpha_{d,N}$. It is the relation between $N$ and $d$ in $\alpha_{d,N}$ that determines when the far-field behavior breaks down. We will go deeper into this in the next section.

\section{Near-Field and Far-Field Distances for Antenna Arrays}
\label{sec:near-far-field}

The previous section uncovered several phenomena that are essential for modeling the total channel gain in the radiative near-field. We noticed that the distance at which we must use the near-field compliant formula is related to the diagonal $\sqrt{2NA}$ of the array, but not to the wavelength. This is somewhat surprising because the Fraunhofer distance, the classic border between the near-field and far-field, is wavelength-dependent. In this section, we will describe the different distances that are related to the radiative near-field and clarify their respective meanings and roles.

\subsection{Phase Variations and Fraunhofer Distance}

The transmit antenna has an aperture consisting of a continuum of point sources, each emitting spherical wave components. The combined wavefront might appear as planar when observed at a sufficiently large distance. The corresponding region is called the far-field and is characterized by the fact that the electric fields induced by the multiple point sources superimpose to create a field strength that is inversely proportional to the propagation distance $d$ with a proportionality constant only depending on the angle \cite{Friedlander}. To derive for which distances this applies, it is instructive to consider the opposite scenario \cite{Selvan}: an isotropic transmit antenna and a receiver aperture with some maximum length $D$. Due to reciprocity, in the far-field, the amplitude of the electric field should be approximately constant over the receive antenna and the phase variations only depend on the incident angle, not the distance.

\begin{figure}[t!]
	\centering
	\begin{overpic}[width=0.9\columnwidth,tics=10]{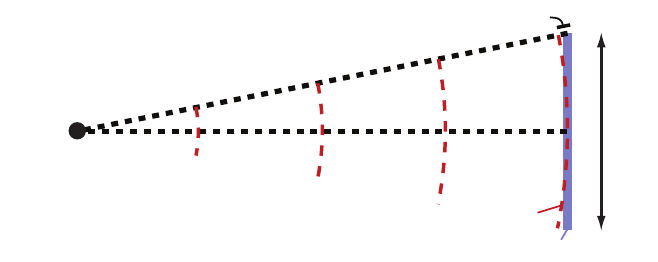}
	 \put (-4,19) {Transmitter}
	 \put (77,1) {Receiver}
	 \put (69,6.5) {Wavefront}
	 \put (51,29) {$d'$}
	 \put (51,21) {$d$}
	 \put (76,37) {$d'\!-\!d$}
	 \put (93,19) {$D$}
\end{overpic} 
	\caption{The curvature of an impinging spherical wave creates a delay $d'\!-\!d$ between the center of the receiver and the edge. The delay turns into a phase-shift of $\frac{2\pi}{\lambda} (d'\!-\!d)$.}
	\label{fig:fraunhofer} 
\end{figure}
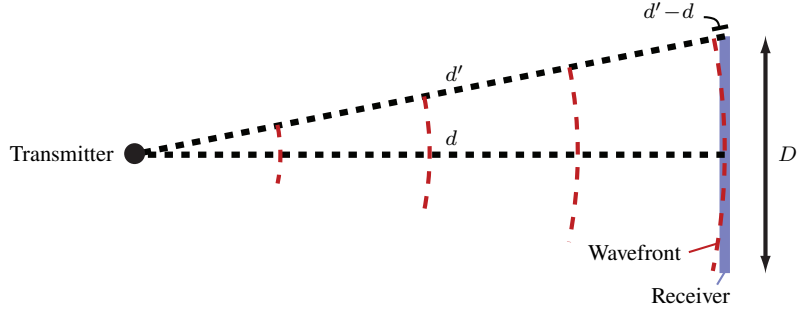

Fig.~\ref{fig:fraunhofer} shows the worst-case scenario where the transmitted wave impinges perpendicularly to the receiver. If the propagation distance to the antenna's center is $d$, then the distance to the edges is $d' = \sqrt{d^2 + (D/2)^2}$.
As illustrated in the figure, the wave needs to travel an additional distance $d'-d$ to reach the edge, which will incur a phase shift relative to the center. Using the angular wavenumber $\frac{2\pi}{\lambda}$ (measured in radians per unit distance), we can compute this phase-shift as
\begin{equation}
\frac{2\pi}{\lambda} \left( d'-d\right) = \frac{2\pi}{\lambda} \left( \sqrt{d^2 + \frac{D^2}{4}} - d\right) \approx 
\frac{2\pi}{\lambda}  \frac{D^2}{8 d},
\end{equation}
where we used the first-order Taylor approximation $\sqrt{1+x} \approx 1+\frac{x}{2}$.

The electromagnetic literature is often treating $\pi/8$ as the maximum phase-shift that can be considered negligible since $\cos(\pi/8)=0.92$ \cite{Selvan}. If we let $d_{\mathrm{F}}$ denote the propagation distance that gives this exact phase-shift, we obtain
\begin{equation} \label{eq:Fraunhofer}
\frac{\pi}{8} = \frac{2\pi}{\lambda}  \frac{D^2}{8 d_{\mathrm{F}}} \quad \rightarrow \quad
d_{\mathrm{F}} = \frac{2D^2}{\lambda}.
\end{equation}
The distance $d_{\mathrm{F}}$ is called the \emph{Fraunhofer distance} \cite{Cheng,Sherman} or \emph{Rayleigh distance} \cite{Kraus2002}.

\begin{example}{Example 2}For an antenna with the maximum length $D= \lambda$, operating 
at the carrier frequency of $f=3$\,GHz (i.e., $\lambda = 0.1$\,m), the Fraunhofer distance is  $d_{\mathrm{F}} = 2\lambda = 0.2$\,m.
\end{example}

This example shows that the Fraunhofer distance is short, even for relatively large antennas. Hence, long-range communication systems operate beyond that distance.

\subsubsection{Fraunhofer Array Distance}
\label{subsec:fraunhofer_array}

An ELAA consists of many antennas, each small compared to the wavelength, but collectively making a large total aperture. Suppose the ELAA is a planar array with $N$ identical antennas, as in Fig.~\ref{fig:geometric_setup}, and each antenna has the diagonal $D$. The maximum length of the array is the diagonal $W=D\sqrt{N}$.
If we require the phase variation of the impinging wave to be less than $\pi/8$ over the array aperture, we need to be beyond the \emph{Fraunhofer array distance} that is obtained by replacing $D$ by $W$ in \eqref{eq:Fraunhofer}:
\begin{equation} \label{eq:Fraunhofer-array}
d_{\mathrm{FA}} = \frac{2 W^2}{\lambda} = \frac{2 \left(D \sqrt{N} \right)^2}{\lambda} = N d_{\mathrm{F}}. 
\end{equation} 
We notice that it is precisely $N$ times larger than the Fraunhofer distance of an individual antenna. The following example shows that $d_{\mathrm{FA}}$ can be very large.

\begin{example}{Example 3}Consider an ELAA that is used at the wavelength $\lambda=0.1$\,m (i.e., $f=3$\,GHz).
The Fraunhofer array distance is $d_{\mathrm{FA}}=20$\,m if the array's diagonal is $W=1$\,m, while it grows to $d_{\mathrm{FA}}= 2$\,km if $W=10$\,m.
If the wavelength shrinks to $\lambda=0.01$\,m (i.e., $f=30$\,GHz), these numbers increase to $200$\,m and $20$\,km, respectively.
\end{example}

This example shows that the Fraunhofer array distance can be very large for ELAAs, such that almost every user served by the array is closer than $d_{\mathrm{FA}}$ to the transmitter. More precisely, we expect propagation distances $d$ such that $d_{\mathrm{F}} \leq d \leq d_{\mathrm{FA}}$. This implies that each antenna observes a locally plane wave in the uplink but the spherical curvature is noticeable when comparing the phases between different antennas. This situation was illustrated in Fig.~\ref{fig:geometric_setup_matched_filtering}, where there are large phase variations over the array aperture, but small variations over the individual antenna.
The phase variations across the array can be compensated for using MF, as demonstrated in Section~\ref{subsec:system-model}, if the channel vector is known.

\subsection{Gain Variations and Bj\"ornson Distance}

One of the main purposes of antenna arrays is to use beamforming (e.g., MF) to achieve a larger total channel gain than with a single antenna. The channel gain grows linearly with the number of antennas in the far-field, while we noticed in Section~\ref{subsec:scaling_law} that the gain saturates for very large arrays. This is due to gain variations that occur over the array in the near-field. We will now take a closer look at these variations from a receiver perspective and assume an isotropic transmitter, but the same result holds in the reciprocal setup with a transmitting array \cite{Sherman,Polk,Kay,Hansen}.

Similarly to Lemma~\ref{lemma1}, we consider an isotropic transmitter located at $\vect{p}_t = (0,0,d)$ that emits a signal with
wavelength $\lambda$ and polarization in the $Y$-direction. If the electric intensity  is denoted as $E_0$ Volts, it follows from \eqref{eq:complex_channel} that the impinging electric
field perpendicular to a receive antenna at location $(x, y, 0)$ is
\begin{equation}
E(x,y)  
=\frac{E_0}{\sqrt{4 \pi}} \frac{\sqrt{d (x^2 + d^2)}}{(x^2+y^2+d^2)^{5/4}} e^{-\imagunit \frac{2\pi}{\lambda}\sqrt{x^2+y^2+d^2}}. \label{eq:intensity-function}
\end{equation} 
This near-field compliant field expression can be compared with the corresponding expression for a plane wave that has the same amplitude and phase at $x=y=0$:
\begin{equation} \label{eq:plane-wave}
E_{\mathrm{plane}}(x,y) = \frac{E_0}{\sqrt{4\pi} } \frac{1}{d} e^{-\imagunit \frac{2\pi}{\lambda}d}.
\end{equation}
We will use these expressions to define the gain and characterize the gain variations.

\subsubsection{Antenna Gain and Antenna Array Gain}

The \emph{antenna gain} determines how effective a particular antenna is compared to an isotropic antenna.
The receive antenna gain for an antenna with area $A$ that spans a subset $\mathcal{S} \subset \mathbb{R}^2$ of the $XY$-plane is defined as \cite[Eq.~(6)]{Kay}
\begin{equation}
G = \frac{ \frac{1}{\eta} \left| \frac{1}{A} \int_{\mathcal{S}}  E(x,y) dx \, dy \right|^2}{\frac{\lambda^2}{4\pi} \frac{1}{\eta} \int_{\mathcal{S}} \frac{1}{A^2} \left| E(x,y) \right|^2 dx \, dy  } = \frac{ \left| \int_{\mathcal{S}}  E(x,y) dx \, dy \right|^2}{\frac{\lambda^2}{4\pi} \int_{\mathcal{S}} \left| E(x,y) \right|^2 dx \, dy  }. \label{eq:G-exact}
\end{equation}
The numerator in the first expression is the received power ($\eta$ is the impedance of free space), obtained similarly to \eqref{eq:total_gain}, divided by the area $\lambda^2/(4\pi)$ of an isotropic antenna multiplied by the average power flow through the antenna aperture. The second expression removes redundant terms.
The largest antenna gain is achieved in the far-field with the perpendicular plane wave defined in \eqref{eq:plane-wave}, which results in 
\begin{equation}
G_{\mathrm{plane}} = \frac{ A^2 \frac{E_0^2}{4\pi d^2 } }{\frac{\lambda^2}{4\pi} A \frac{E_0^2}{4\pi d^2 }} = \frac{4 \pi A}{\lambda^2}. \label{eq:G-plane}
\end{equation}
To measure how close to the maximum gain we can reach in the near-field, we define the \emph{normalized antenna gain} as the ratio between \eqref{eq:G-exact} and \eqref{eq:G-plane}:
\begin{equation}
G_{\mathrm{antenna}} =  \frac{G}{G_{\mathrm{plane}}} =  \frac{ \left| \int_{\mathcal{S}}  E(x,y) dx \, dy \right|^2}{A \int_{\mathcal{S}} \left| E(x,y) \right|^2 dx \, dy  }. \label{eq:G-exact-normalized}
\end{equation}

Next, we will define the \emph{antenna array gain} for a planar ELAA with $N$ antennas. Recall that we defined the antenna coordinates $\vect{p}_n=(x_n,y_n,0)$ for $n=1,\ldots,N$ in 
\eqref{eq:xn} and \eqref{eq:yn}.
This implies that antenna $n$ covers the area
\begin{align} 
\mathcal{S}_n = \left\{ (x,y) : \left| x - x_n \right| \leq \frac{a}{2},  \left| y - y_n  \right| \leq \frac{a}{2} \right\} \subset \mathbb{R}^2.
\end{align}
Similar to \eqref{eq:G-exact-normalized}, the normalized antenna array gain can be defined as 
\begin{equation} \label{eq:antenna-array-gain-exact}
G_{\mathrm{array}} = \frac{  \sum_{n=1}^N \left|  \int_{\mathcal{S}_n}  E(x,y) dx \, dy \right|^2 }{ N A  \int_{\mathcal{S}} \left| E(x,y) \right|^2 dx \, dy},
\end{equation}
which is the total received power of the $N$ antennas divided by the received power of $N$ references antennas. The reference antenna captures all the power of the electric field and is located in the origin with $\mathcal{S} = \{ (x,y) : | x  | \leq \frac{a}{2},  | y  | \leq \frac{a}{2} \}$.
This definition from \cite{Bjornson4} combines the antenna gains and array gain into a single metric. These gains are normally treated as separate multiplicative factors since they decouple in the far-field but we need to treat them jointly in the near-field.
The normalized antenna array gain $G_{\mathrm{array}} $ takes values between $0$ and $1$. The maximum value is preferred since the array will then capture the same power as in the ideal case of having an incident wave with perpendicular planar wavefronts.
We can either compute \eqref{eq:antenna-array-gain-exact} numerically or
use the upper bound from Corollary~\ref{cor:alpha_expression} as follows.

\begin{corollary} \label{cor:normalized-array-gain}
The normalized array gain in \eqref{eq:antenna-array-gain-exact} can be upper bounded as
\begin{align} 
& G_{\mathrm{array}}  \leq  \frac{\alpha_{d,N} }{ N \alpha_{d,1}},
\label{eq:G-upper-bound}
\end{align}
where $\alpha_{d,N} $ is defined in \eqref{eq:alpha-expression}.
\end{corollary}
\begin{proof}
The proof is given in Appendix~\ref{app:proof_corr15}.
\end{proof}

The upper bound in \eqref{eq:G-upper-bound} is tight for electrically small antennas (e.g., $a \leq \lambda/4$) for the same reason as in Theorem~\ref{lemma2}; that is, $E(x,y)$ is constant over each antenna.

\subsubsection{Bj\"ornson Distance}

The normalized antenna array gain in \eqref{eq:antenna-array-gain-exact} is close to $1$ when the propagation distance is sufficiently large. This occurs for distances beyond the \emph{Bj\"ornson distance} \cite{Bjornson2,Bjornson4}
\begin{equation} \label{eq:Bjornson}
d_{\mathrm{B}} = 2 W = 2D \sqrt{N},
\end{equation}
where $W=\sqrt{2NA}$ is the diagonal of the ELAA and $D=\sqrt{2A}$ is the diagonal of an antenna.
We noticed the same thing in Fig.~\ref{fig:near-field_modeling}, where the near-field and far-field curves begin to deviate precisely when the propagation distance is twice the largest dimension $W$ of the antenna.
The expression in \eqref{eq:Bjornson} differs substantially from the Fraunhofer array distance in \eqref{eq:Fraunhofer-array}: it grows with the number of antennas as $\sqrt{N}$ instead of $N$, and is wavelength-independent since the free-space decay in amplitude is the same for all frequencies. Consequently, we typically have $d_{\mathrm{FA}} \geq d_{\mathrm{B}}$, a condition that can be rewritten as
\begin{equation}
\frac{2 \left(D \sqrt{N} \right)^2}{\lambda}  \geq 2 D \sqrt{N}  \,\,\, \Rightarrow \,\,\, N \geq \frac{\lambda^2}{D^2}.
\end{equation}
This inequality holds for $N\geq 1$ if $D=\lambda$ and for $N \geq 16$ if $D=\lambda/4$. Moreover, the Bj\"ornson distance is larger than $d_{\mathrm{F}}$ whenever $N \geq (D/\lambda)^2$.

The selection of $d_{\mathrm{B}}$ can be motivated geometrically by limiting the loss in received power due to the spherical wavefront.
We will exemplify it by considering a worst-case circular array with diameter $W$: $S_{\mathrm{circ}} = \{ (x,y) : x^2+y^2 \leq (W/2)^2 \}$.
Similar to Fig.~\ref{fig:fraunhofer}, we consider a transmitter at distance $d$ and compute the ratio of the received power of the antenna with a spherical wavefront and with a planar wavefront:
\begin{equation}
    \frac{\int_{S_{\mathrm{circ}}} \frac{1}{d^2+x^2+y^2} dx dy }{ \int_{S_{\mathrm{circ}}} \frac{1}{d^2} dx d\theta}
    = \frac{\pi \log \left( 1 +  \left( \frac{W}{2} \right)^2 \frac{1}{d^2}\right)}{\pi \left( \frac{W}{2} \right)^2 \frac{1}{d^2}} \approx 1 - \frac{W^2}{8d^2},
\end{equation}
where we used the second-order Taylor approximation $\log(1+x) \approx x-\frac{x^2}{2}$. The second term is the relative power loss and it becomes
$1/32\approx 0.03$ for $d=d_{\mathrm{B}}=2W$.
At this distance, the average power variations are negligible over the array. Note that this metric differs from classical ones in the electromagnetic literature that focus on keeping the maximum amplitude variation below $\cos(\pi/8)$ \cite{Sherman}, which has less operational meaning in communications.

\begin{figure}[t!]
        \centering

	\begin{overpic}[width=\columnwidth,tics=10]{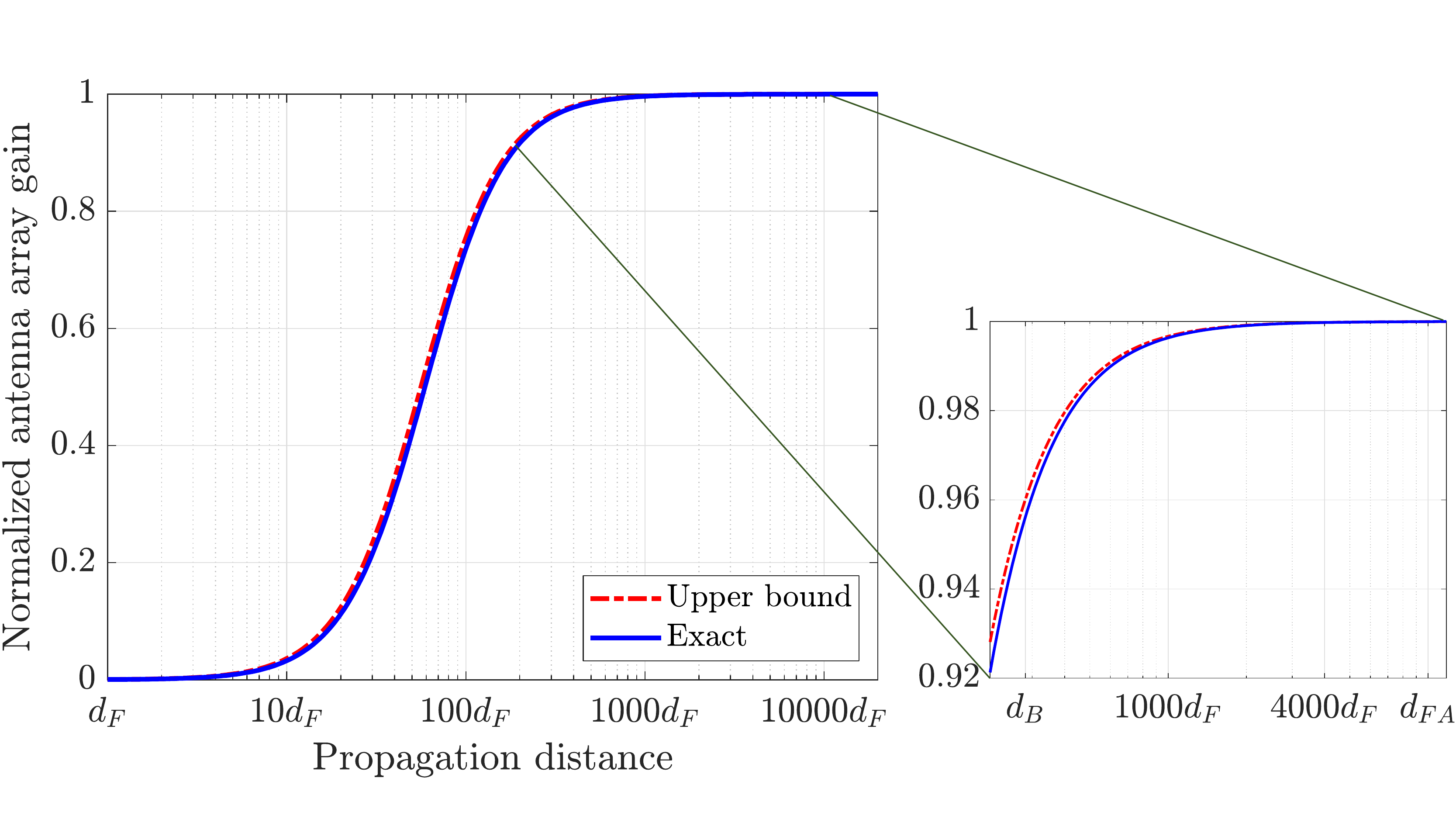}
\end{overpic}  \vspace{-4mm}

        \caption{The normalized antenna array gain $G_{\mathrm{array}} $ is about $0.96$ for $z \geq d_{\mathrm{B}}$.}
        \label{fig:normalized_gain}
\end{figure}

Fig.~\ref{fig:normalized_gain} shows the normalized antenna array gain $G_{\mathrm{array}} $ as a function of the distance $d$ from the transmitter to the center of the planar ELAA. There are $N=10^4$ antennas, each with the area $A=(\lambda/4)^2$.
The Fraunhofer distance becomes $d_{\mathrm{F}} = \lambda/4$ and is used as the reference unit on the horizontal axis, to make the results wavelength-independent. We compare the exact $G_{\mathrm{array}}$ in \eqref{eq:antenna-array-gain-exact} with the upper bound in \eqref{eq:G-upper-bound}. 
We notice that around 96\% of the maximum gain is achieved for $d \geq d_{\mathrm{B}} \approx 283 d_{\mathrm{F}}$ (as expected from the Bj\"ornson distance), while the normalized gain is very close to $1$ at the Fraunhofer array distance $d_{\mathrm{FA}}=10^4 d_{\mathrm{F}}$. There is a large interval between these distances because $d_{\mathrm{FA}}/d_{\mathrm{B}} \approx 35$.
We further notice that the bound in \eqref{eq:G-upper-bound} is tight since each antenna is electrically small.

The example confirms that we (almost) achieve the maximum antenna array gain whenever the receiver is at a distance $d \geq d_{\mathrm{B}}$ (larger than the Bj\"ornson distance), even if $d \leq d_{\mathrm{FA}}$.
The reason is that each antenna is in the far-field of the transmitter \cite{Friedlander} and, thus, observes a locally plane wave from which it can extract the maximum gain.
In the practical scenario of $d_{\mathrm{B}} \leq d \leq d_{\mathrm{FA}}$, the spherical curvature is noticeable when comparing the local phases between the $N$ antennas, as illustrated in Fig.~\ref{fig:geometric_setup_matched_filtering}, and MF compensates for it.
This would not be possible with a single antenna having the same total area $NA$, thus demonstrating a key benefit of antenna arrays.

\subsection{Finite-Depth Near-Field Beamforming} \label{subsec:nearfield_beamforming}

Although the Fraunhofer array distance is not determining the achievable antenna array gain $G_{\mathrm{array}} $, it manifests what kind of signal processing is needed to achieve it.
If the propagation distance is $d \geq d_{\mathrm{FA}}$, we can use the classic
plane wave approximation and the MF will use the far-field array response vector that only depends on the angle of arrival/departure.
In contrast, for $d \leq d_{\mathrm{FA}}$, we need to consider the spherical curvature and the MF will depend both on the angle and distance.
Moreover, in this section, we will show that the transmitted beam might also have a limited depth-of-focus (DF). 

When the MF is designed to focus the transmitted signal at a point $(0,0,z)$, the maximum antenna array gain is obtained at that point and smaller numbers at other points $(0,0,d)$ in the same angular direction.
We want to quantify the distance dependence of the gain.
To this end, we begin by computing an approximate analytical expression for $G_{\mathrm{array}} $ in \eqref{eq:antenna-array-gain-exact} using the classic Fresnel approximation $E(x,y) \approx \frac{E_0}{\sqrt{4\pi} z} e^{-\imagunit \frac{2\pi}{\lambda} (z + \frac{x^2}{2z} + \frac{y^2}{2z})}$ of the electric field \cite[Eq.~(22)]{Polk}:
\begin{align} \notag
G_{\mathrm{array}} &\approx \left( \frac{1}{A N} \right)^2 \sum_{n=1}^N \left| e^{-\imagunit \frac{2\pi}{\lambda} z  }\int_{\mathcal{S}_n}  e^{-\imagunit \frac{2\pi}{\lambda} (\frac{x^2}{2z} + \frac{y^2}{2z})} dx \, dy  \right|^2 \\ \notag
&= \left( \frac{1}{A N} \right)^2 \left| \int_{-\sqrt{NA}/2}^{\sqrt{NA}/2}  e^{-\imagunit \frac{\pi}{\lambda} \frac{x^2}{z}} dx \right|^4 \\
&= \left( \frac{8 z}{d_{\mathrm{FA}}} \right)^2 \left( C^2 \left( \sqrt{\frac{d_{\mathrm{FA}}}{8 z} }\right) \!+\! S^2 \left( \sqrt{\frac{d_{\mathrm{FA}}}{8 z} }\right) \right)^2 \label{eq:G-approx},
\end{align}
where $C(x) = \int_{0}^{x} \cos(\pi t^2/2) dt$ and $S(x) = \int_{0}^{x} \sin(\pi t^2/2) dt$ are the Fresnel integrals, which can be evaluated using the error function. 

The antenna array gain observed at another point $(0,0,d)$ can be evaluated as 
\begin{align} \notag
G_{\mathrm{array},d} &\approx \left( \frac{1}{A N} \right)^2 \sum_{n=1}^N \left| e^{-\imagunit \frac{2\pi}{\lambda} d  }\int_{\mathcal{S}_n}  e^{-\imagunit \frac{2\pi}{\lambda} (\frac{x^2}{2d} + \frac{y^2}{2d})} e^{+\imagunit \frac{2\pi}{\lambda} (\frac{x^2}{2z} + \frac{y^2}{2z})} dx \, dy  \right|^2 \\ 
&= \left( \frac{8 z_{d}^{\textrm{eff}}}{d_{\mathrm{FA}} } \right)^2 \left( C^2 \left( \sqrt{\frac{d_{\mathrm{FA}} }{8 z_{d}^{\textrm{eff}}} }\right) + S^2 \left( \sqrt{\frac{d_{\mathrm{FA}}}{8 z_{d}^{\textrm{eff}}} }\right) \right)^2 \label{eq:G-approx-DF},
\end{align}
where $z_{d}^{\textrm{eff}} = \frac{dz}{|d-z|}$ represents the focal point deviation. 
Note that \eqref{eq:G-approx-DF} is computed similarly to \eqref{eq:G-approx} except that $z$ is replaced with $d$, and the phase-shift $e^{+\imagunit \frac{2\pi}{\lambda} (\frac{x^2}{2z} + \frac{y^2}{2z})}$ is injected into the integral to represent the effect of MF with small antennas.

We can define the DF as the distance interval $d \in [z_{\min},z_{\max}]$ where the antenna array gain is at most $3$\,dB lower than the maximum value \cite{Sherman,Nepa} (achieved at the focal point $d=z$).
We notice that  $G_{\mathrm{array},d}$ has the structure  $A(x) = (C^2(\sqrt{x}) + S^2(\sqrt{x}) )^2 /x^2$, where $x = d_{\mathrm{FA}}  / (8 z_{d}^{\textrm{eff}})$. Moreover, $A(x) $ is a decreasing function for $x \in [0,2]$ with $A(0) =1$ and $A(1.25) \approx 0.5$.  Hence, the $3$\,dB gain loss is obtained when 
\begin{equation}
1.25 = \frac{Nd_{\mathrm{F}}}{8 z_{d}^{\textrm{eff}}} = \frac{Nd_{\mathrm{F}} |d-z|}{8 dz}  \,\, \rightarrow \,\, d = \frac{d_{\mathrm{FA}} z  }{d_{\mathrm{FA}}   \pm 10 z}.
\end{equation}

\begin{figure}[t!]
\begin{center}
	\begin{overpic}[trim={5mm 5mm 12mm 5mm},width=0.85\columnwidth,tics=10]{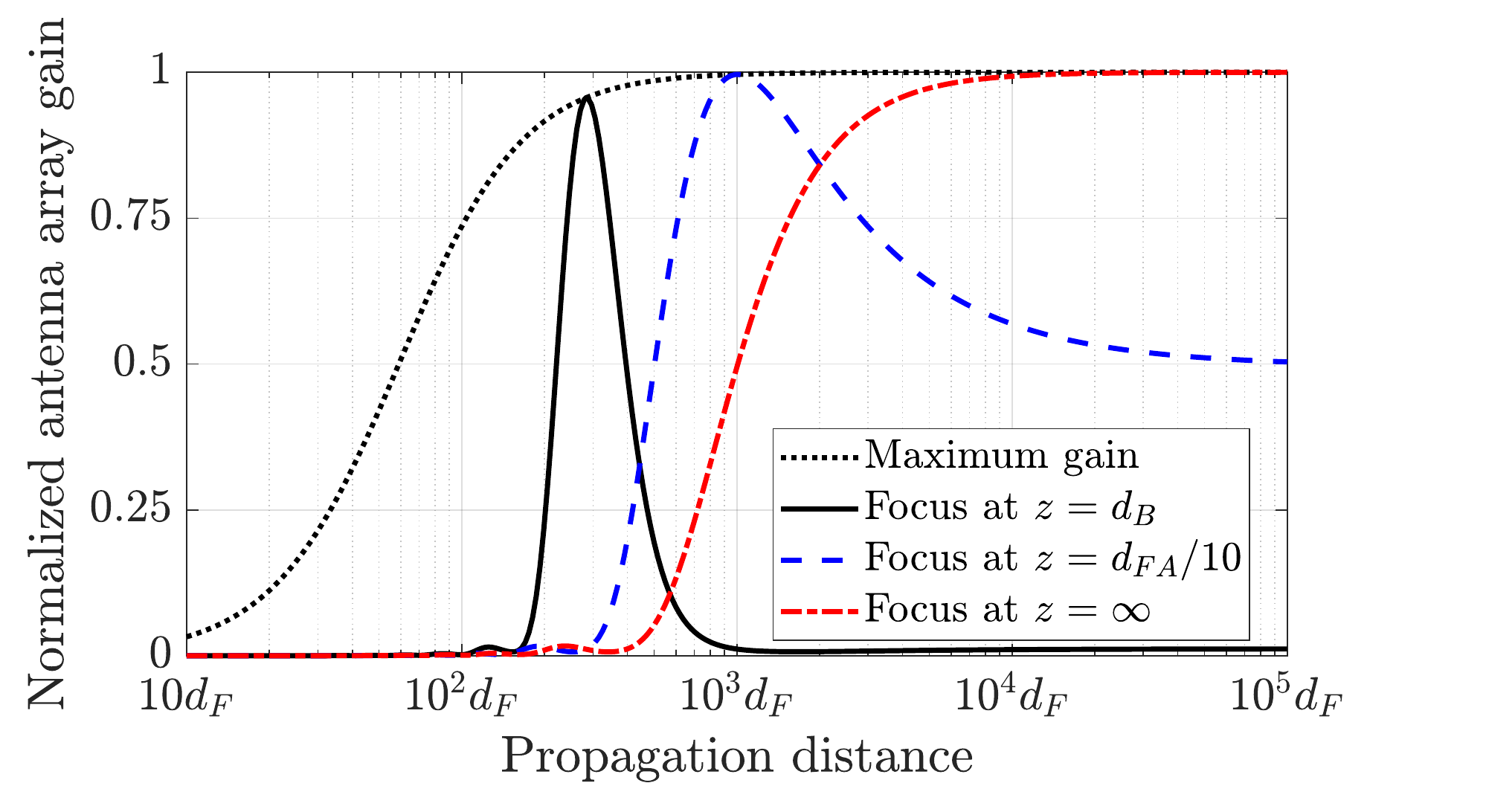}
	 \put(71.6,46.5){\vector(0,1){4}}
	 \put(69,44){\footnotesize $d_{\mathrm{FA}}$}
	 \put(51.55,46.5){\vector(0,1){4}}
	 \put(49,42.5){\footnotesize $\frac{d_{\mathrm{FA}}}{10}$}
	 \put(38.6,53.1){\footnotesize $d_{\mathrm{B}}$}
	 \put(40.4,51.7){\vector(0,-1){2}}
\end{overpic} 
\end{center} 
\caption{The depth of a beam depends on the location the MF focuses on. The depth is finite for near-field beamforming, when the focal point is closer than $d_{\mathrm{FA}}/10$.} \label{fig:focusing}   \vspace{-2mm}
\end{figure}
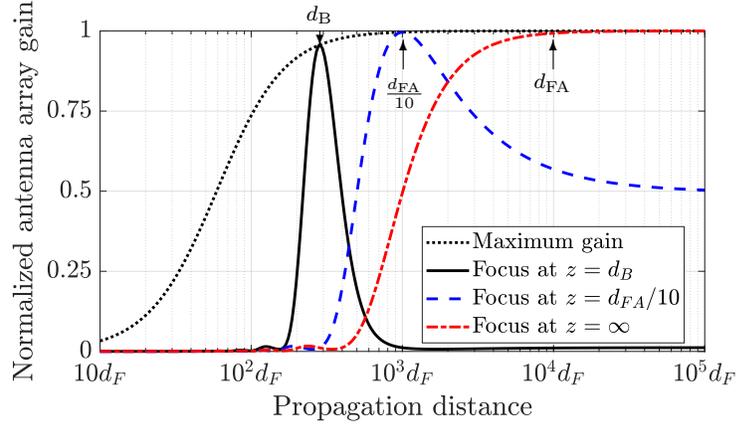

\begin{theorem} \label{th:DF}
When MF is utilized to focus on a receiver at $(0,0,z)$, the 3\,dB depth-of-focus in the same direction is
\begin{equation}
d \in \left[ \frac{d_{\mathrm{FA}} z }{d_{\mathrm{FA}}  + 10 z} ,  \frac{d_{\mathrm{FA}} z  }{d_{\mathrm{FA}}  - 10 z} \right] \label{eq:BD-interval}
\end{equation}
if $z < d_{\mathrm{FA}} /10$. Otherwise, the upper limit in \eqref{eq:BD-interval} is replaced by $\infty$.
\end{theorem}
Fig.~\ref{fig:focusing} shows the maximum normalized gain achieved at different distances from an array with $N=100^2 = 10^4$ antennas, each with length $A = (\lambda/4)^2$.
The figure also shows how signals arriving from different distances are amplified when the matched filtering is selected to focus on a transmitter located at three different distances.
The Fraunhofer distance is $d_{\mathrm{F}} = \lambda/4$, while $d_{\mathrm{B}} \approx 283 d_{\mathrm{F}}$ and $d_{\mathrm{FA}} = 10^4 d_{\mathrm{F}}$.
For far-field focusing at $z=\infty$, the normalized gain is between $1$ and $0.5$ (i.e., $-3$\,dB) in the interval $[d_{\mathrm{FA}}/10,\infty)$. This is the DF and is in line with Theorem~\ref{th:DF}.
If the matched filtering focuses on $z = d_{\mathrm{FA}}/10$, the DF is $[ d_{\mathrm{FA}}/20,\infty)=[ 500 d_{\mathrm{F}} ,\infty)$ but we approach $-3$\,dB as $d \to \infty$.
When the focal point is $z = d_{\mathrm{B}}$, the DF interval is roughly $[220 d_{\mathrm{F}}, 394 d_{\mathrm{F}}]$, which is rather narrow. As expected, there is also a noticeable loss in maximum gain when the focal point is the Bj\"ornson distance.

The result in Theorem~\ref{th:DF} and the observations of Fig.~\ref{fig:focusing} uncover another key property of the radiative near-field:
when focusing on a receiver closer than $d_{\mathrm{FA}} /10$, the beam will have limited DF. 
For more distant focal points, the $3$\,dB beam depth (BD) extends to infinity, as expected from conventional far-field beamforming. Moreover, as the focal point $z \to \infty$, the lower limit in \eqref{eq:BD-interval} approaches $d_{\mathrm{FA}} /10$, thus making it a natural border between near-field and far-field beamforming. We will call the length of the interval in \eqref{eq:BD-interval} the \emph{3\,dB BD} and it can be computed as
\begin{align}
    \textrm{BD}_{\textrm{3\,dB}} = \begin{cases} \frac{d_{\mathrm{FA}} z  }{d_{\mathrm{FA}}  - 10 z}-\frac{d_{\mathrm{FA}} z  }{d_{\mathrm{FA}}  + 10 z} = \frac{20 d_{\mathrm{FA}} z^2}{d_{\mathrm{FA}}^2 - 100z^2}, & z < \frac{d_{\mathrm{FA}}}{10}, \\
    \infty,  & z \geq \frac{d_{\mathrm{FA}}}{10}.
    \end{cases} \label{eq:BD}
\end{align}
The distinction between the two cases in \eqref{eq:BD} is illustrated in Fig.~\ref{fig:beamdepth}, where the gain is large and the color is strong. A conventional far-field beam (Fig.~\ref{fig:beamdepth1}) begins at roughly the distance $d_{\mathrm{FA}}/10$ and then continues towards infinity, while a near-field beam (Fig.~\ref{fig:beamdepth2}) has a finite depth around the focal point. Hence, the first illustration of beamforming that we provided in Fig.~\ref{fig:mu-MIMO-basic} is oversimplified because beams can never span the entire distance from the transmitter to infinity. With an ELAA, we can make use of both features depending on the propagation distance, similar to how the same camera can take close-ups with a blurry background (finite DF) and landscape photos that are sharp from a certain distance to infinity.
The dotted lines in Fig.~\ref{fig:beamdepth} illustrate how the angular beamwidth is the same irrespective of the focal distance \cite{Bjornson4}, while the beamwidth in meters depends on the focal distance.

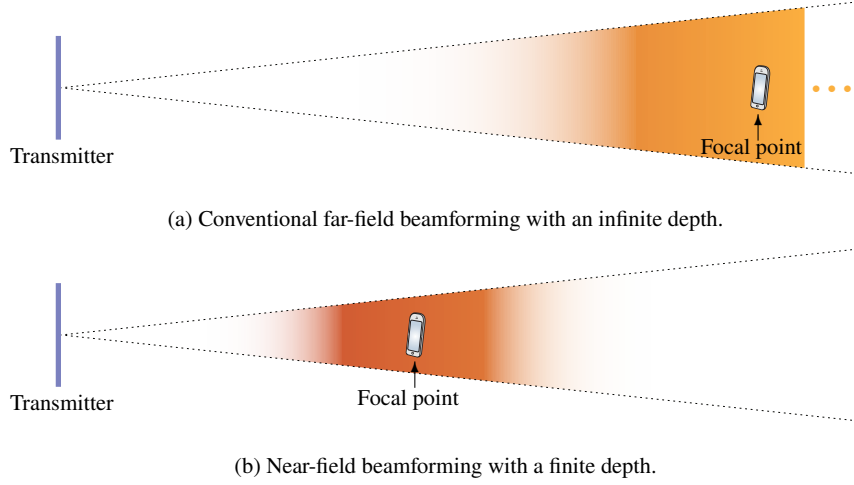
\begin{figure}[t!]
        \centering
        \begin{subfigure}[b]{\columnwidth} \centering 
	\begin{overpic}[width=\columnwidth,tics=10]{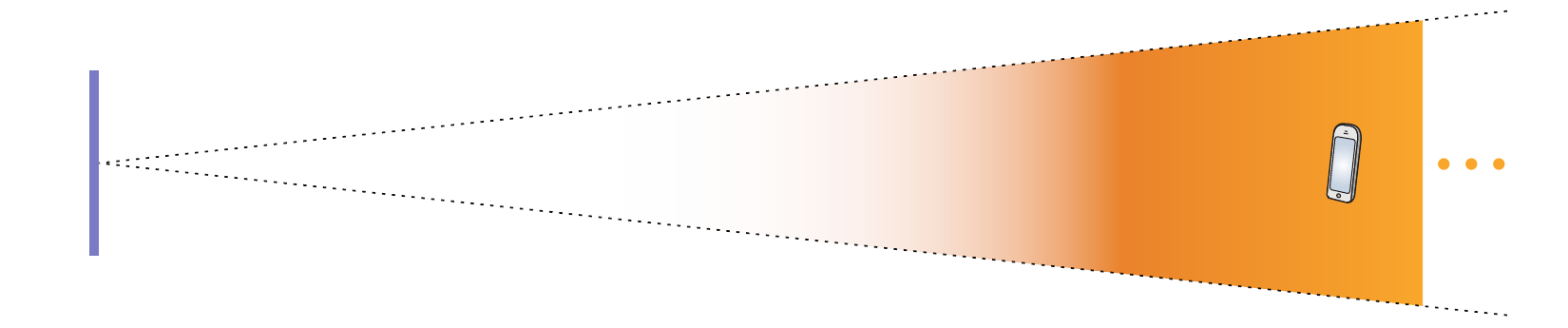}
 	\put(0.5,2){Transmitter}
 	\put(79,3){Focal point}
	 \put(85.5,4.5){\vector(0,1){3}}
\end{overpic}  \vspace{-2mm}
                \caption{Conventional far-field beamforming with an infinite depth.} 
            \label{fig:beamdepth1}
        \end{subfigure}\\
        \begin{subfigure}[b]{\columnwidth} \centering  \vspace{+2mm}
	\begin{overpic}[width=\columnwidth,tics=10]{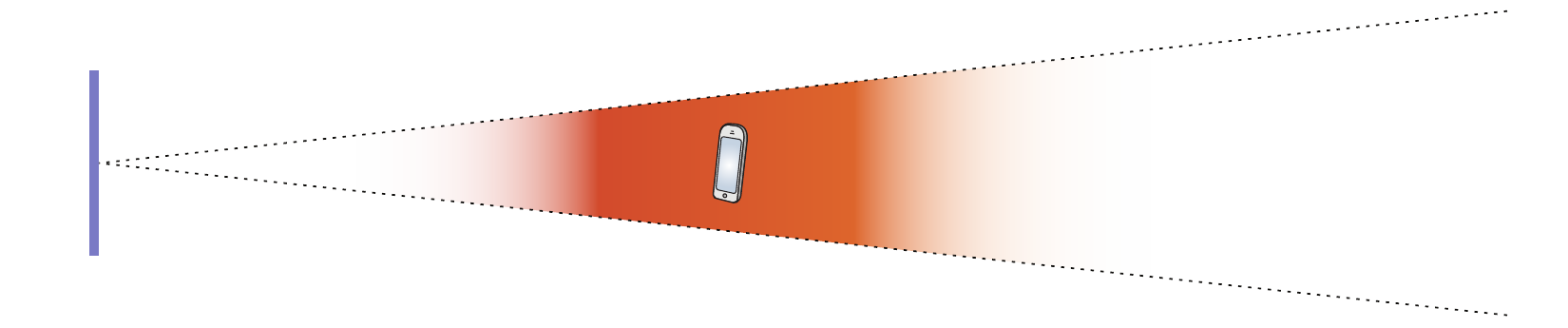}
 	\put(0.5,2){Transmitter}
 	\put(46.5,4.5){\vector(0,1){3}}
 	\put(40,2.5){Focal point}
\end{overpic}  \vspace{-2mm}
                \caption{Near-field beamforming with a finite depth.} 
                     \label{fig:beamdepth2}
        \end{subfigure} 
        \caption{The depth of a beam depends on whether it is focused on a receiver located in the far-field ($z \geq d_{\mathrm{FA}}/10$) or radiative near-field ($z < d_{\mathrm{FA}}/10$).}
        \label{fig:beamdepth}
\end{figure}

\section{Near-Field Multiplexing in the Depth Domain}

The near-field beamforming described in Section \ref{subsec:nearfield_beamforming} enables the use of depth as a new dimension for communication.
In this section, we will demonstrate how we can utilize the depth domain to increase the capacity of a system by serving multiple users that are located in the same angular direction with respect to the ELAA, but at sufficiently different distances.

Spatial multiplexing in the far-field relies on serving users located at sufficiently different angles, in the sense that the beamwidths of their main lobes from the access point are non-overlapping.
We will apply the same principle for multiplexing users in the depth domain.
From Theorem~\ref{th:DF}, we know that the DF when using MF to focus a beam in the far-field (i.e.,  $d_1=\infty$) begins at the distance $d_{\mathrm{FA}}/10$. We can now compute a second focal point $d_2$ that has its upper limit of the 3\,dB BD interval exactly where the far-field beam begins:
\begin{equation}
  \frac{d_{\mathrm{FA}}d_2}{d_{\mathrm{FA}}-10d_2} = \frac{d_{\mathrm{FA}}}{10} \, \, \Rightarrow \, \, d_2 = \frac{d_{\mathrm{FA}}}{20}.  
\end{equation}
The DF interval for $d_2 = d_{\mathrm{FA}}/20$ is $[d_{\mathrm{FA}}/30,~ d_{\mathrm{FA}}/10]$. Similarly, we can obtain the focal point $d_3 = d_{\mathrm{FA}}/40$, which has the DF interval  $[d_{\mathrm{FA}}/50,~d_{\mathrm{FA}}/30]$
for which the upper limit matches with the lower limit for $d_2$.
The next two focal points will be $d_4 = d_{\mathrm{FA}}/60$ and $d_4 = d_{\mathrm{FA}}/80$.

\begin{figure}[t!]
	\centering 
	\begin{overpic}[trim={5mm 5mm 12mm 5mm},width=0.85\columnwidth,tics=10]{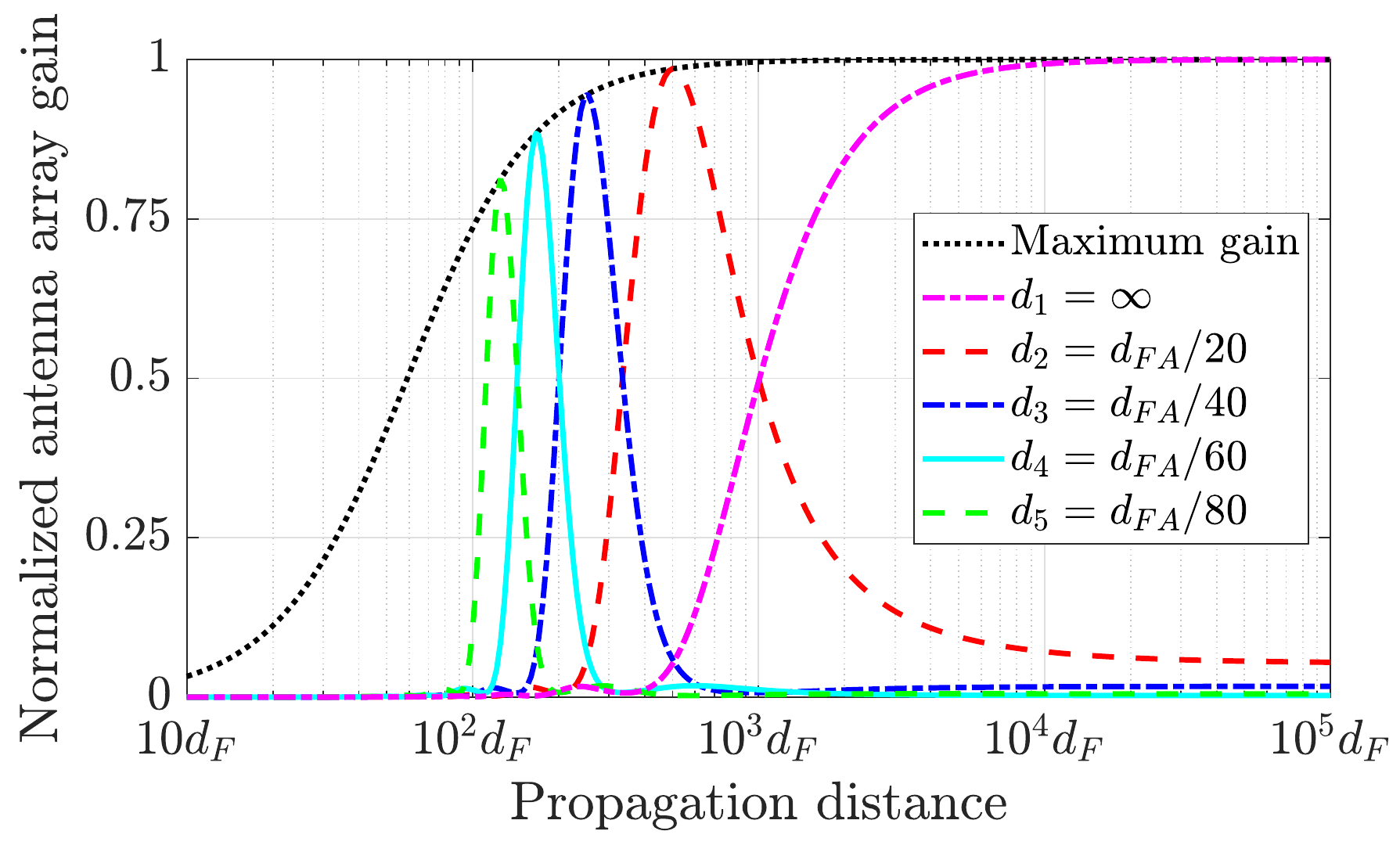}
\end{overpic} 
	\caption{The finite depth of near-field beamforming enables the ELAA to focus multiple beams in the same direction but at different distances. The five focal points are selected so that the consecutive DF intervals intersect where the normalized antenna array gain is $0.5$.}
	\label{fig:simulation_nearfield_multiplexing}  
\end{figure}

Fig.~\ref{fig:simulation_nearfield_multiplexing} shows the normalized antenna array gains obtained at different distances when MF beamforming is used to transmit at the aforementioned five focal points. We consider the same setup as in Fig.~\ref{fig:focusing} and notice that the focal points were selected so that the consecutive DF intervals intersect where the normalized gain is $0.5$.
Since these DF intervals are clearly different, the ELAA can serve all five users simultaneously even if they are in the same angular direction.

The channel vector $\vect{h}  = [h_1,\ldots,h_N]^{\Ttran}$ to a focal point at the arbitrary distance $d$ can be computed as described in Section~\ref{subsec:system-model}, for the case of electrically small antennas. 
A more general expression can be obtained by integration of the electric field in \eqref{eq:intensity-function}:
\begin{equation}
h_{n} = \frac{1}{E_0} \sqrt{\frac{1}{A}} \int_{\mathcal{S}_{n}}  E(x,y) dx \, dy.
\end{equation}
By generating the channel vectors to multiple focal points, we can apply standard methods to evaluate the SE with different multi-user beamforming schemes \cite{massivemimobook}.

We now continue the example in Fig.~\ref{fig:depth_multiplexing}, where we use the five beams to transmit downlink data to users located at the respective focal points. The only change is that we use $d_1=d_{\mathrm{FA}}$ since the channel gain is zero at infinity. The figure shows the sum SE that is achieved for different SNR values. Since the users experience widely different channel gains, the reference SNR is achieved at the outermost user when it is allocated all the power.
We compare depth-domain spatial multiplexing of the users, based on zero-forcing (ZF) beamforming and waterfilling power allocation, with a scheduling baseline where each user is allocated $1/5$ of the time resources.
We notice that ZF achieves a roughly $4\times$ higher SE than scheduling, which validates the ability to distinguish between users in the depth domain using near-field beamforming.
Interestingly, the user closest to the ELAA achieves the highest SE, despite having the lowest gain in Fig.~\ref{fig:simulation_nearfield_multiplexing}, because it has the largest channel gain. It is the shape of the beams that determines the ability to spatially distinguish users, while the channel gain determines the final SE.

\begin{figure}[t!]
	\centering 
	\begin{overpic}[width=0.9\columnwidth,tics=10]{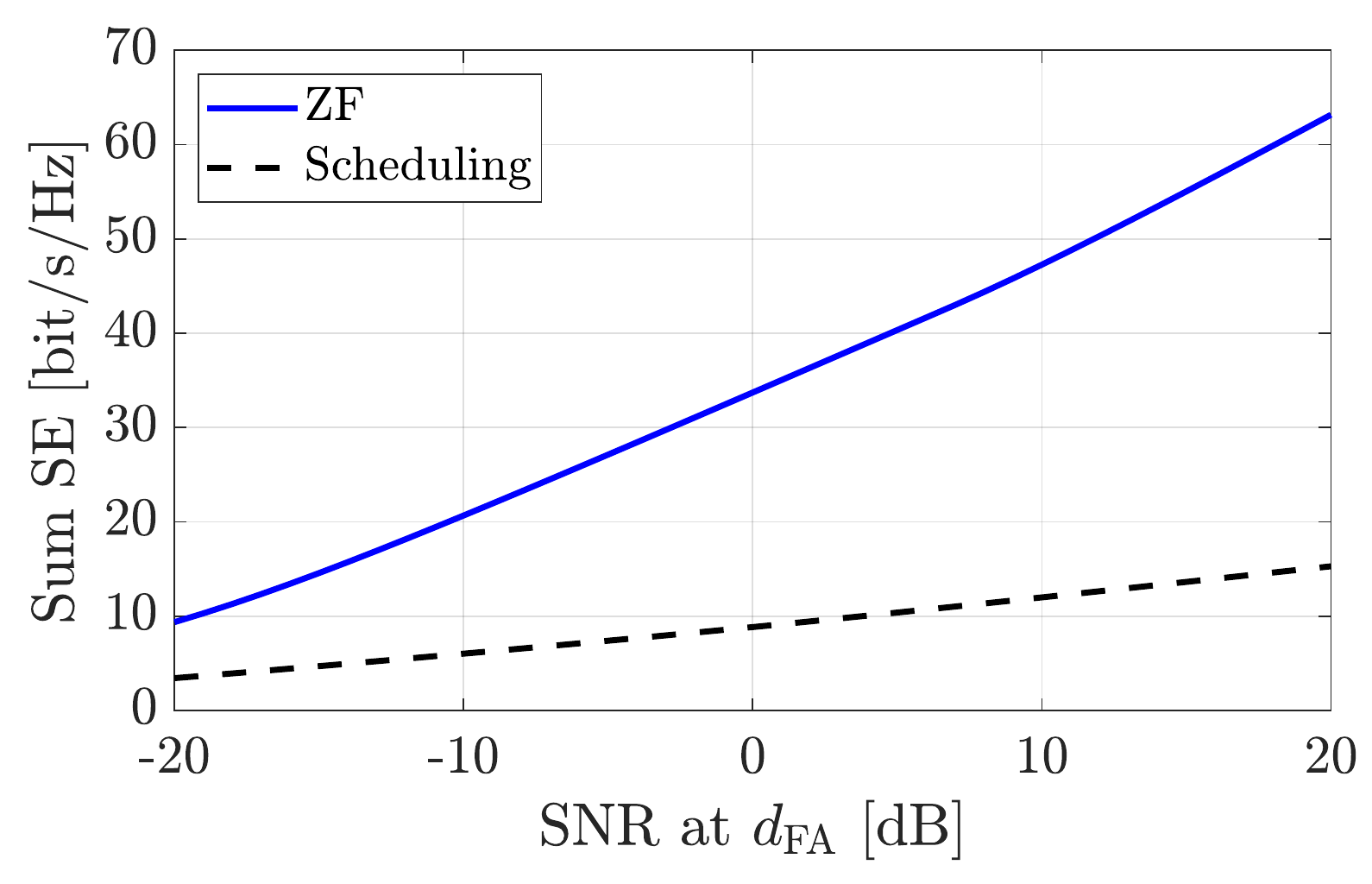}
		 \put(88,53){\vector(2,1){8}}
	 \put(64,54){Slope: $5 \log_2(\mathrm{SNR})$}
\end{overpic} 
	\caption{The sum SE achieved by depth-domain spatial multiplexing of five users located in the same angular direction, but at different propagation distances. ZF is compared with a scheduling baseline where the users take turns.}
	\label{fig:depth_multiplexing}  
\end{figure}

The curve in Fig.~\ref{fig:depth_multiplexing} grows as $K \log_2( \mathrm{SNR}) + \mathrm{constant}$ at high SNR, where $K=5$ is the number of multiplexed users. This is a classical scaling behavior for multi-antenna systems \cite{Lozano2005a} and demonstrates that the $K$ channel vectors span a $K$-dimensional vector space. The factor $K$ is called the \emph{spatial degrees-of-freedom (DoF)} and manifests the maximum capacity that the system can achieve. 
When considering an ELAA with thousands of antennas, it is interesting to quantify the maximum DoF, to determine how many users could be efficiently served by spatial multiplexing.
It is proved in \cite{Hu,Pizzo2020a} that this limit for a large planar array is
\begin{equation}
    K \leq \pi \cdot \frac{\textrm{Array area}}{\lambda^2},
\end{equation}
which implies that each segment of area $\lambda^2$ that is added to the array enables spatial multiplexing of another $\pi$ users.
Importantly, it is not the number of antennas or the antenna spacing that determines the DoF, but the array's aperture. This is why we specifically need ELAAs in future systems, while the benefits of building them using electrically small antennas only provide a comparably minor SNR boost.
We can keep up with an increasing user load by adding extra segments to the ELAA.

\begin{example}{Example 4}The maximum spatial DoF with an ELAA of size $1.79 \times 1.79$\,m is $1000$ at the wavelength $\lambda=0.1$\,m (i.e., $f=3$\,GHz) and $10^5$ at $\lambda=0.01$\,m (i.e., $f=30$\,GHz). Hence, we can achieve huge DoF with relatively small arrays, since the aperture is the physical size relative to the wavelength.
\end{example}

\section{Conclusion}

There are two fundamental ways to increase the capacity that a wireless access point can deliver: 1) Use larger antenna apertures to enable stronger beamforming and more spatial multiplexing; 2) Use more bandwidth at shorter wavelengths.
Both approaches lead to a paradigm where the user devices will predominately be in the radiative near-field of the access point. In this chapter, we have analyzed ELAAs that operate in LoS scenarios. Firstly, we demonstrated that classical metrics such as the Fraunhofer distance $d_{\mathrm{F}}$ and the natural Fraunhofer array distance extension $d_{\mathrm{FA}}$ are unsuitable for characterizing the border between the near-field and far-field in communications.
As illustrated in Fig.~\ref{fig:distances-ELAA}, the radiative near-field instead begins close to the array and continues roughly until the distance $d_{\mathrm{FA}}/10$. In this region, the spherical wavefronts are noticeable and must be taken into account. Secondly, we demonstrated that beamforming in the radiative near-field results in beams with a finite depth.
The depth domain can be a game changer when multiplexing massive crowds of users, which would otherwise be hard to distinguish by the access point.

\begin{figure}[t!]
\begin{center}
	\begin{overpic}[width=\columnwidth,tics=10]{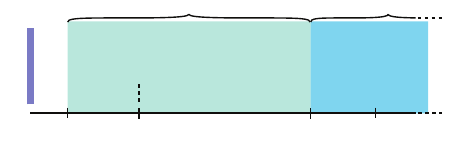}
	 \put(1,27){Transmitter}
	 \put(14,2){$d_{\mathrm{F}}$}
	 \put(29.5,2){$d_{\mathrm{B}}$}
	 \put(66.5,2){$\frac{d_{\mathrm{FA}}}{10}$}
	 \put(81,2){$d_{\mathrm{FA}}$}
	 
	 \put(33.5,22){Spherical waves}
	 \put(32.5,16){Finite beam depth}
	 
	 \put(15.7,10){Reduced gain}
	 \put(41,10){Maximum gain}
	 
	 \put(74.5,22){Plane waves}
	 \put(70,16){Beam depth to infinity}
	 \put(73,10){Maximum gain}
	 \put(27,30){\textbf{Radiative near-field region}}
	 \put(77,30){\textbf{Far-field region}}
\end{overpic} 
\end{center} \vspace{-6mm}
\caption{A summary of the main distinctions between the radiative near-field and the far-field, when these regions are defined based on the communication characteristics.} \label{fig:distances-ELAA} 
\end{figure}
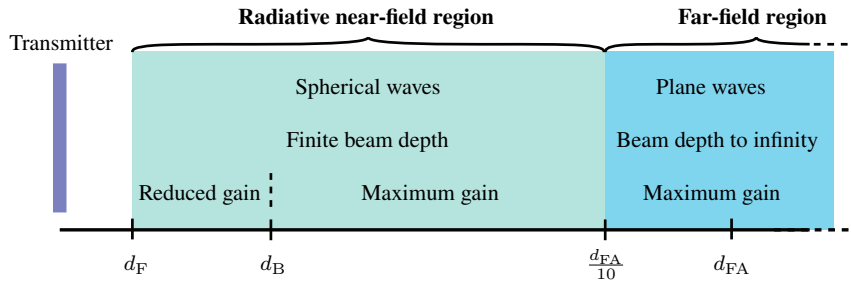

If the beamforming scheme (e.g., MF) takes the spherical waves into account, we can often achieve the same antenna array gain as in the far-field. However, when closer than the Bj\"ornson distance $d_{\mathrm{B}}$, the gain will also decrease due to three essential near-field characteristics: varying antenna distances, effective antenna areas, and polarization losses. It is mainly the channel models (and potentially the channel estimation protocols \cite{Wei2022a}) that need to be revised when considering ELAAs, while the signal processing methods and capacity analysis from mMIMO textbooks such as \cite{massivemimobook} remain valid.
The main open research challenges are instead related to signal processing and hardware design \cite{Amiri2018a,Zhang2022a} (e.g., cost-efficient implementation of ZF with many antennas and users), channel modeling for more complex scenarios than LoS \cite{Pizzo1,Demir,Dong2022a}, and prototyping.

\backmatter
\appendix

\chapter{Proofs of Main Results in Chapter 1}

\section{Proof of Lemma \ref{lemma1}}
\label{app:proof_lemma1}

Consider an elementary transmitting surface with area $A_t$ and centroid located in $\vect{p}_t = [x_t,y_t,d]$.
The electric field ${\bf E}(\vect{p}_t, \vect{p}_r)\in \mathbb{C}^{3}$ generated at a point $\vect{p}_r = [x_r,y_r,0]$ takes the form \cite{Dardari}
\begin{align}
{\bf E}(\vect{p}_t, \vect{p}_r)  = {\bf G} (\vect{p}_r - \vect{p}_t){\bf J}(\vect{p}_t)
\end{align}
where ${\bf J}(\vect{p}_t) = J_x(\vect{p}_t)\hat{{\bf u}}_x+ J_y(\vect{p}_t)\hat{{\bf u}}_y + J_z(\vect{p}_t)\hat{{\bf u}}_z$ is the \emph{radiation vector}, which is measured in [A$\cdot$m] and is determined by the surface's current density. Note that $\hat{{\bf u}}_x,\hat{{\bf u}}_y,\hat{{\bf u}}_z$ represent the unit vectors
in the $x,y,z$ directions.

The Green function ${\bf G} (\vect{p}_r)\in \mathbb{C}^{3\times 3}$ is
well-approximated for $|| \vect{p}_r|| \gg \lambda$ as \cite{Poon}
\begin{align}\label{eq:green_function}
{\bf G} (\vect{p}_r) = -j \eta\frac{e^{-j \frac{2\pi }{\lambda}{||\vect{p}_r||}}}{2\lambda ||\vect{p}_r||}\left({\bf I}_3 - \hat{{\bf p}}_r\hat{{\bf p}}_r^{\Htran}\right)
\end{align}
with $\hat{{\bf p}}_r =  \frac{\vect{p}_r}{|| \vect{p}_r||}$ and $\eta$ is the impedance of free space. This approximation is tight when the receiver is beyond the reactive near-field of the transmitter.

It was assumed in the lemma that only the $Y$-direction of ${\bf J}(\vect{p}_t)$ is excited at the transmitter, thus we have that ${\bf J}(\vect{p}_t) = J_y(\vect{p}_t)\hat{{\bf u}}_y$. The electric field reduces to 
\begin{align}
{\bf E}(\vect{p}_t, \vect{p}_r)  & = {\bf G}_y (\vect{p}_r - \vect{p}_t)J_y(\vect{p}_t) 
\end{align}
where ${\bf G}_y (\vect{p}_r  - \vect{p}_t)={\bf G} (\vect{p}_r - \vect{p}_t)\hat{{\bf u}}_y$ is the second column of the Green function in \eqref{eq:green_function}. The complex-valued channel coefficient $\epsilon(\vect{p}_t, \vect{p}_r)$ from the considered transmitting surface located in $\vect{p}_t$ to the receive point $\vect{p}_r$ in the $XY$-plane can be divided into its amplitude and phase as
\begin{align}\label{eq:channel}
\epsilon(\vect{p}_t, \vect{p}_r) = |\epsilon(\vect{p}_t, \vect{p}_r)|e^{-j \frac{2\pi }{\lambda}{||\vect{p}_r - \vect{p}_t||}}.
\end{align}
It follows from \cite[Eqs. (16) and (19)]{Dardari} that
\begin{align} \notag
|\epsilon(\vect{p}_t, \vect{p}_r)|^2 &=  A_t \overbrace{\frac{4}{\eta^2}||{\bf G}_y (\vect{p}_r - \vect{p}_t)||^2}^{\text{Power gain}}\overbrace{\frac{ ({\vect{p}_r - \vect{p}_t})^{\Ttran}\hat{{\bf u}}_z}{||{\vect{p}_r - \vect{p}_t}||}}^{\text{Projection on the $Z$ direction}} \\ &= \frac{1}{4 \pi}{\frac{ d \left( (x_r-x_t)^2 + d^2 \right)}{ \left( (x_r -x_t)^2 + (y_r-y_t)^2 + d^2 \right)^{5/2}}}.\label{eq:channel_gain}
\end{align}
As indicated on the first row, this is the channel gain in the $Z$-direction (i.e., perpendicularly to the array) where $\frac{ {\vect{p}_r - \vect{p}_t}}{||{\vect{p}_r - \vect{p}_t}||}$ is the pointing direction of the electric field and $A_t = \lambda^2/(4\pi)$ is the area of an isotropic antenna. 
The considered antenna is assumed to have the dimensions $a  \times a$ in the $XY$-plane, around $\vect{p}_n=(x_n,y_n,0)$, thus the channel is
\begin{align}
h_n(\vect{p}_t) =  \frac{1}{a} \int_{x_n-a/2}^{x_n+a/2} \int_{y_n-a/2}^{y_n+a/2} \epsilon(\vect{p}_t, \vect{p}_r) \partial x_r \partial y_r.
\end{align}
The channel gain can be computed as $| h_n(\vect{p}_t) |^2$ and is given in \eqref{eq:total_gain}.

\section{Proof of Theorem \ref{lemma2}} \label{app:proof_lemma2}

We can compute an upper bound on the channel gain expression in \eqref{eq:total_gain} as
\begin{align} \notag
\left|h_n(\vect{p}_t)\right|^2 &= \left|   \frac{1}{a} \int_{x_n-a/2}^{x_n+a/2} \int_{y_n-a/2}^{y_n+a/2} \epsilon(\vect{p}_t, \vect{p}_r) \partial x_r \partial y_r\right|^2 \\ &\leq \int_{x_n-a/2}^{x_n+a/2} \int_{y_n-a/2}^{y_n+a/2} \left|\epsilon(\vect{p}_t, \vect{p}_r)\right|^2 \partial x_r \partial y_r   = \zeta_{\vect{p}_t,\vect{p}_n},\label{eq:channel_gain_approximated}
\end{align}
 using the Cauchy-Schwarz inequality 
$ |  \iint  \epsilon(\vect{p}_t, \vect{p}_r) dx  dy |^2 \!\leq \!   \iint  |\epsilon(\vect{p}_t, \vect{p}_r)  |^2 dx  dy \cdot \iint 1 dx  dy$.

To compute $\zeta_{\vect{p}_t,\vect{p}_n}$ in \eqref{eq:channel_gain_approximated} in closed form, we need to solve the integral

\begin{align} 
\zeta_{\vect{p}_t,\vect{p}_n} &= \frac{1}{4\pi} \int_{x_n-a/2}^{x_n+a/2} \int_{y_n-a/2}^{y_n+a/2} \frac{ d \left( (x_r-x_t)^2 + d^2 \right)   \partial x_r \partial y_r}{ \left( (x_r -x_t)^2 + (y_r-y_t)^2 + d^2 \right)^{5/2} } \label{eq:pathloss-integral} \\ \notag 
&= \int_{x_n-a/2}^{x_n+a/2} \int_{y_n-a/2}^{y_n+a/2}
\underbrace{ \frac{ d }{ \sqrt{ (x_r -x_t)^2 + (y_r-y_t)^2 + d^2 } }}_{\textrm{Reduction in effective area from directivity}} \\ \notag &\times 
\underbrace{\frac{  (x_r-x_t)^2 + d^2  }{ (x_r -x_t)^2 + (y_r-y_t)^2 + d^2  }}_{\textrm{Polarization loss factor}} \times \underbrace{\frac{  \partial x_r \partial y_r }{ 4\pi ( (x_r -x_t)^2 + (y_r-y_t)^2 + d^2 ) } }_{\textrm{Free-space pathloss}{}}.
\end{align}
The contributions of the three fundamental properties when operating in the near-field of the array (i.e., the variations in distances to the antennas, in the effective antenna areas, and in the polarization losses) are stated explicitly in this expression. 
The rest of the proof follows from computing the integral in \eqref{eq:pathloss-integral} and the details can be found in \cite[App.~A]{Bjornson2}.

\vspace{-1mm}

\section{Proof of Corollary~\ref{cor:far-field_mMIMO}}
\label{app:proof_cor:far-field_mMIMO}

When $d \cos(\varphi) \gg \sqrt{2N A}$, it follows that $B+1 \approx 1$ and $2B+1 \approx 1$. We can then utilize the fact that $\tan^{-1}(x) \approx x$ for $x \approx 0$  to approximate \eqref{eq:xi-mMIMO} as
\begin{align} \xi_{d,\varphi,N} \approx \sum_{i=1}^{2} \frac{ B +(-1)^i \sqrt{B} \tan(\varphi)  }{2 \pi \sqrt{ \tan^2(\varphi) + 1 + 2(-1)^i \sqrt{B} \tan(\varphi)} }.  \label{eq:xi-mMIMO-approx1}
\end{align}
Furthermore, we can utilize the fact that $\sqrt{1+x} \approx 1 + x/2$ for $x\approx 0$ to approximate the denominator of \eqref{eq:xi-mMIMO-approx1} and obtain
\begin{align} \notag
\xi_{d,\varphi,N} & \approx \sum_{i=1}^{2} \frac{ B +(-1)^i \sqrt{B} \tan(\varphi)  }{2 \pi \sqrt{1+\tan^2(\varphi)} \left( 1 + \frac{(-1)^i \sqrt{B} \tan(\varphi)}{1+\tan^2(\varphi)} \right)} \\ \notag
&
= \frac{ 2B - \frac{2B \tan^2(\varphi) }{1+\tan^2(\varphi)} }{2 \pi \sqrt{1+\tan^2(\varphi)} \left( 1 + \frac{ \sqrt{B} \tan(\varphi)}{1+\tan^2(\varphi)} \right) \left( 1 - \frac{ \sqrt{B} \tan(\varphi)}{1+\tan^2(\varphi)} \right)} \\
& \approx \frac{B}{\pi (1+\tan^2(\varphi))^{3/2}} = N  \underbrace{\beta_{d \cos(\varphi)} \cos^3(\varphi)}_{=\zeta_{d,\varphi}}
\label{eq:xi-mMIMO-approx2}
\end{align}
where we simplified the expression by writing the two fractions as a single fraction, then utilized that $1 - \frac{(-1)^i \sqrt{B} \tan(\varphi)}{1+\tan^2(\varphi)} \approx 1$, and finally that $1+\tan^2(\varphi) = 1/\cos^2(\varphi)$.

\vspace{-1mm}

\section{Proof of Corollary~\ref{cor:normalized-array-gain}}
\label{app:proof_corr15}

The upper bound is obtained by applying the Cauchy-Schwarz inequality to the numerator of \eqref{eq:antenna-array-gain-exact} as $ |  \int_{\mathcal{S}_n}  E(x,y) dx  dy |^2 \!\leq \!   \int_{\mathcal{S}_n}  |E(x,y)  |^2 dx  dy \int_{\mathcal{S}_n}1 dx  dy$. This results in
\begin{align} 
& G_{\mathrm{array}}  \leq  \frac{ \sum_{n=1}^N   A \int_{\mathcal{S}_n}  \left| E(x,y) \right|^2 dx \, dy  }{ N  A    \int_{\mathcal{S}} \left| E(x,y) \right|^2 dx \, dy}  =  \frac{  \int_{-\sqrt{NA}/2}^{\sqrt{NA}/2}  \int_{-\sqrt{NA}/2}^{\sqrt{NA}/2}  \left| E(x,y) \right|^2 dx \, dy  }{  N  \int_{-\sqrt{A}/2}^{\sqrt{A}/2}  \int_{-\sqrt{A}/2}^{\sqrt{A}/2} \left| E(x,y) \right|^2 dx \, dy}  \label{eq:G-upper-bound-derivation}.
\end{align}
The remaining integrals are of the same kind as in 
Theorem~\ref{lemma2} and equals $\alpha_{d,N}$ and $\alpha_{d,1}$ from Corollary~\ref{cor:alpha_expression}, respectively.

\end{document}